\documentclass[11pt,letterpaper]{article}
\usepackage[top=1in, bottom=1in, left=1in, right=1in]{geometry}
\usepackage{amsfonts}
\usepackage{xspace}
\usepackage{amsmath, amsthm, amssymb}
\usepackage{subcaption}
\usepackage{graphicx}
\usepackage{url}
\usepackage{hyperref}
\usepackage{algorithm,algpseudocode}
\usepackage{graphicx}
\usepackage{subcaption}

\newtheorem{lemma}{Lemma}

\newtheorem{proposition}{Proposition}
\newtheorem{theorem}{Theorem}
\theoremstyle{definition}
\newtheorem{definition}{Definition}

\title{Distributed Transactions: Dissecting the Nightmare}
\author{
Diego Didona\protect\footnote{\'Ecole Polytechnique F\'ed\'erale de Lausanne, IC, Station 14, CH-1015, Lausanne, Switzerland} $^,$\footnote{firstname.lastname@epfl.ch}
\and
Rachid Guerraoui\protect$^,$\footnotemark[1] $^,$\footnotemark[2]
\and
Jingjing Wang\protect$^,$\footnotemark[1] $^,$\footnotemark[2]
\and
Willy Zwaenepoel\protect$^,$\footnotemark[1] $^,$\footnotemark[2]
}

\date{}

\begin{document}

\maketitle

\thispagestyle{empty}

\begin{abstract}
Many distributed storage systems are transactional and a lot of work has been devoted to optimizing their performance, especially the performance of  read-only transactions that are considered the most frequent in practice. Yet, the results obtained so far are rather disappointing, and some of the design decisions seem contrived. This paper contributes to explaining this state of affairs by  proving intrinsic limitations of transactional storage systems, even those that need not ensure strong consistency but only causality.  

We first consider general storage systems where some transactions are read-only and some also involve write operations. We show that even   read-only transactions cannot be ``fast'': their operations  cannot be executed within one round-trip message exchange between a client seeking an object and the server storing it. 
We then consider systems (as sometimes implemented today) where all transactions are read-only, i.e., updates are performed as individual operations outside transactions. In this case,  read-only transactions can indeed be ``fast'', but we prove that they need to be ``visible''. They induce inherent updates on the servers, which in turn impact their overall performance.
\end{abstract}

\section{Introduction}

Transactional distributed storage systems have proliferated in the last decade: Amazon's Dynamo \cite{dynamo}, Facebook's Cassandra \cite{Lakshman_cassandra_2010}, Linkedin's Espresso \cite{qiao_brewing_2013},  Google's Megastore \cite{baker_cidr_2011}, Walter \cite{sovran_transactional_2011} and Lynx \cite{zhang_transaction_2013}
are seminal examples, to name a few. 
A lot of effort 
has been devoted to optimizing their performance for their  success heavily relies  on  their ability to execute  transactions in a fast manner  \cite{speed}. 
Given the difficulty of the task, 
two major ``strategic'' decisions have been made. The first is to prioritize  \emph{read-only transactions}, which allow clients to read multiple items at once from a consistent view of the data-store. Because many workloads are read-dominated, optimizing the performance of read-only transactions has been considered of primary importance.
The second is the departure from strong consistency models \cite{herlihy_linearizability_1990, lamport_interprocess_1986} towards weaker ones \cite{brewer_conjecture_2000, gilbert_brewer_2002, lu_snow_2016, lipton_pram_1988, attiya_sequential_1994, mavronicolas_linearizable_1999}. Among such  weaker consistency models, \emph{causal consistency} has garnered a lot of attention for it avoids heavy  synchronization inherent to strong consistency, can be implemented in an always-available fashion in geo-replicated settings (i.e., despite partitions), while providing sufficient semantics for many applications \cite{lloyd_settle_2011, lloyd_stronger_2013, bailis_bolt-on_2013, du_gentlerain_2014, zawirski_write_2015, akkoorath_cure_2016, mehdi_occult_2017}.
Yet, even the performance of highly optimized state-of-the-art causally consistent transactional storage systems has revealed disappointing. 
In fact, the benefits and implications  of many designs are unclear, and their overheads with respect to systems that provide no consistency are not well understood.

To illustrate this situation, we report here on two state-of-the-art  designs. The first   implements what we call ``fast'' read-only transactions. They  complete in one round of interaction between a client seeking to read the value of an object and the server storing it. This design is implemented by the recent COPS-SNOW \cite{lu_snow_2016}  system, which however makes the assumption that write operations are supported only outside the scope of a transaction.\footnote{Under this assumption, single-object write and a transaction that only writes to one object are equivalent.} The second design implements ``slow'' read-only transactions, that require   two communication rounds to complete. In particular, we consider the design of Cure  \cite{akkoorath_cure_2016}, which supports generic read-write transactions. 
We compare these two systems with three read-dominated workloads corresponding to 0.999, 0.99 and 0.95 read-write ratios, where clients perform read-only transactions and single-object write operations in closed loop.\footnote{We implemented these in the same C++ code-base using Google Protobuf library for communication. We run the workload on a 10Gbps Ethernet network using 64 AMD Opteron 6212 machines with 8 physical cores (16 hardware threads) and 64 GB of RAM (where 32 machines host client processes, and 32 host server processes).} 
Figure 1 reports on the average  latency of read-only transactions for the two designs (``fast'' and ``slow'') as a function of the delivered throughput, and compares them with those achieved by a system that guarantees no consistency (``no'').
The plots depict two results.
First, the slow case results in a 
higher latency with respect to a design with only one round and no consistency. 
This raises the question whether it is possible to preserve the rich semantics of generic read-write transactions and implement read-only transactions with a single communication round.
Second, the performance achieved by the fast read-only transactions are worse than the ones achieved by the slow ones, both in latency and throughput, even for read-write ratios as low as 5\%.
This is unexpected.

\begin{figure}[!h]
    \centering
  \begin{subfigure}[b]{0.32\textwidth}
    \includegraphics[width=\textwidth]{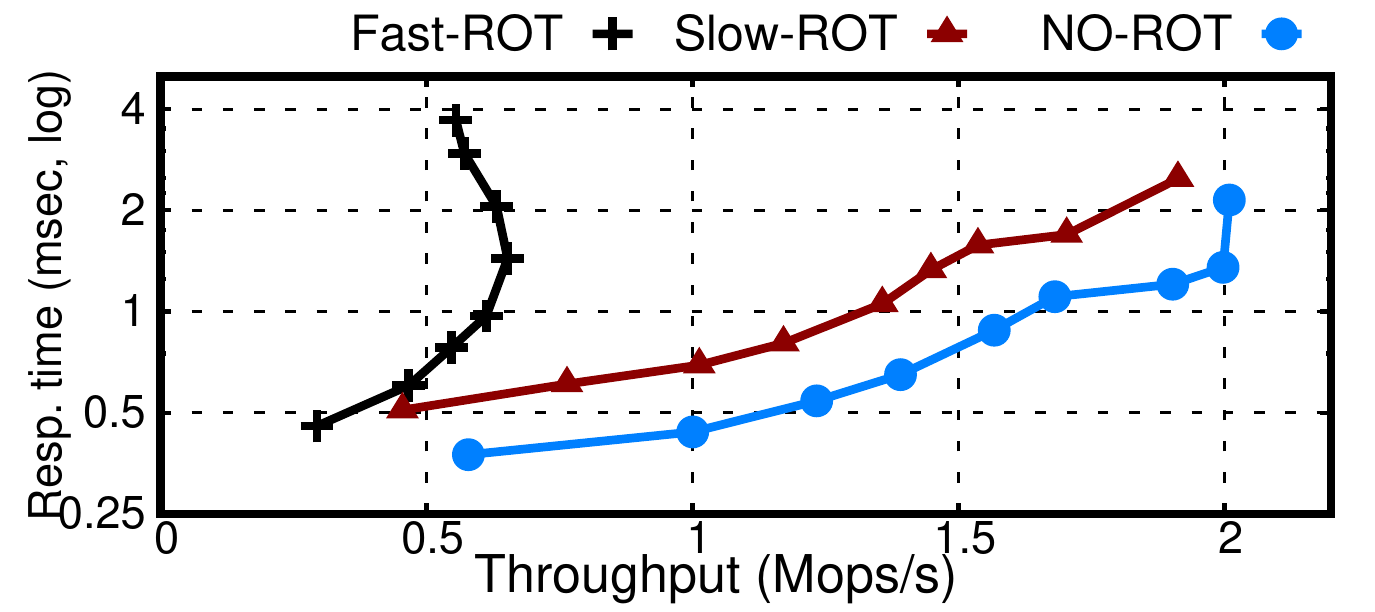}
    \caption{5\% writes}
    \end{subfigure}
  \begin{subfigure}[b]{0.32\textwidth}
    \includegraphics[width=\textwidth]{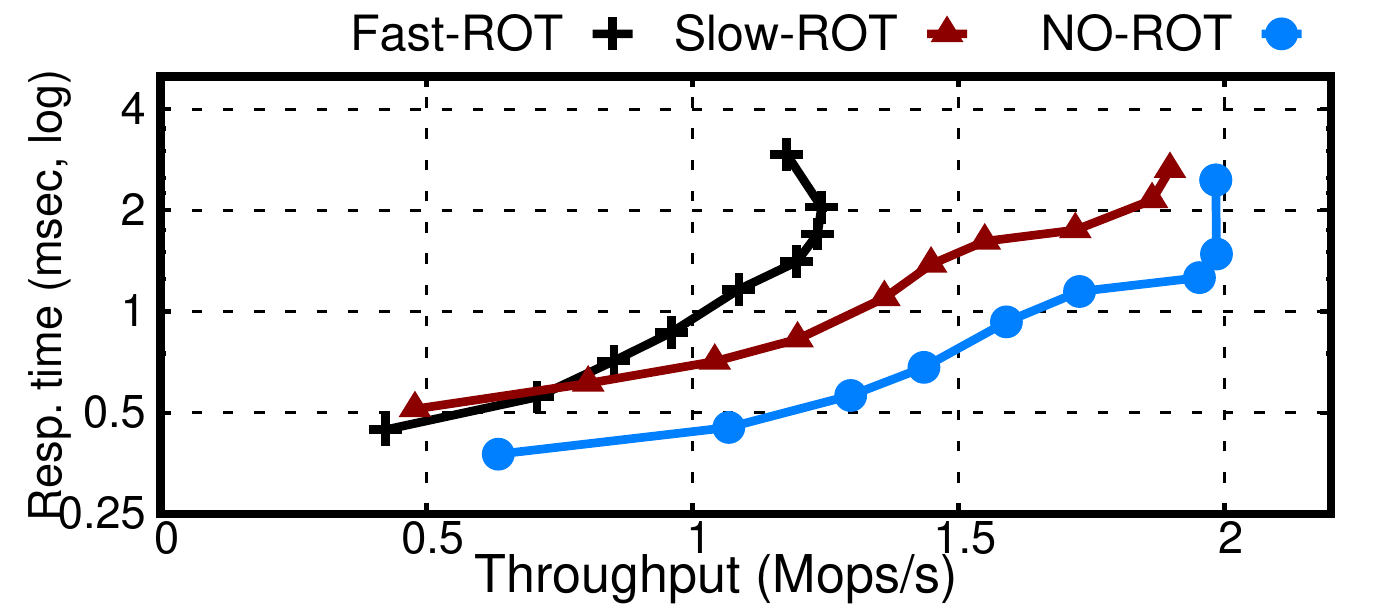}
    \caption{1\% writes}
    \end{subfigure}
  \begin{subfigure}[b]{0.32\textwidth}
    \includegraphics[width=\textwidth]{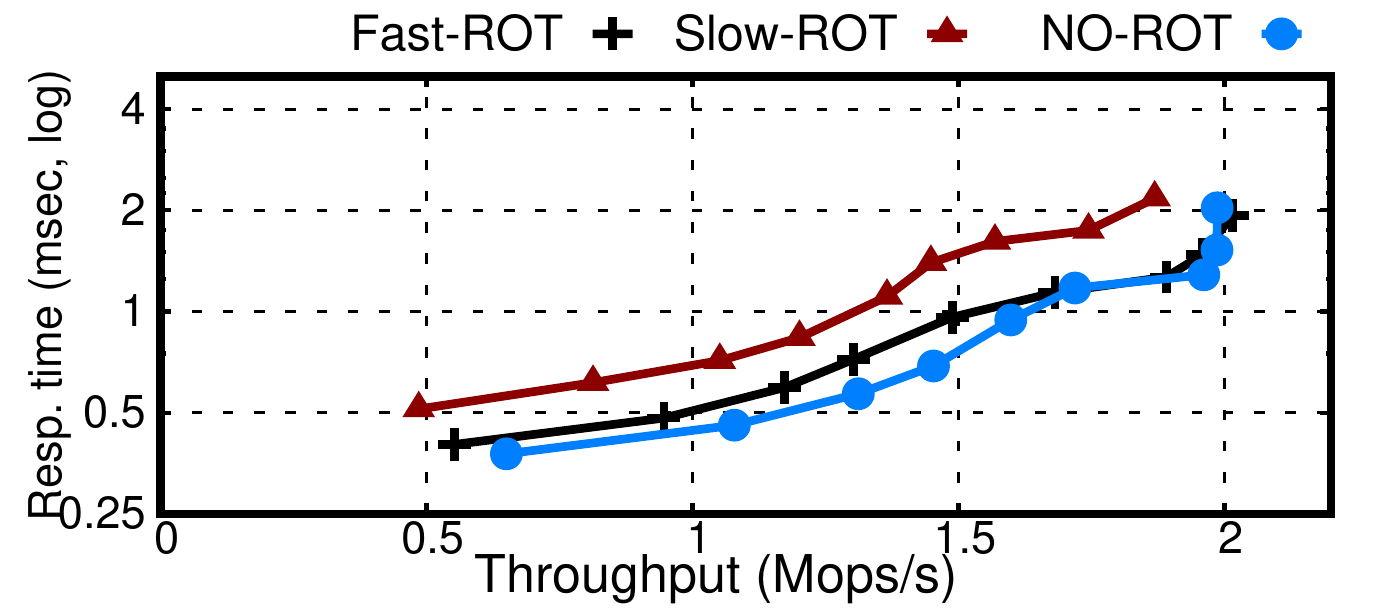}
    \caption{0.1\% writes}
    \end{subfigure}
    \caption{Performance of ``fast'', ``slow'' transactions, and no transaction guarantees}
    \label{fig:1dc}
\end{figure}

In this paper, we investigate these  aspects from a theoretical perspective with the aim of identifying possible and impossible causal consistency designs in order to ultimately understand their implications. We prove two impossibility results. 
\begin{itemize}
\item First, we prove that no causally consistent system can support generic transactions and implement fast read-only transactions. This result unveils  a fundamental trade-off between semantics (support for generic transactions) and performance (latency of read-only transactions).

\item Second, we prove that fast read-only transactions must be ``visible'', i.e., their execution updates the states of the involved servers. 
The resulting overhead   increases resource utilization, which sheds light on the inherent overhead of fast read-only transactions and explains the surprising result discussed before.
\end{itemize}

The main idea behind our first  impossibility result is the following. One round-trip message exchange disallows multiple servers to synchronize their responses to a client. Servers need to be conservative and return possibly stale values to the client in order to preserve causality, with the risk of jeopardizing  progress.
Servers have no choice but  communicate outside read-only transactions (i.e., helping each other) to make progress on the freshness of values. 
We show that such message exchange can cause an infinite loop and delay fresh values forever.
The  intuition  behind our second result is different. We show that a fast read-only transaction  has to ``write'' to some server for otherwise,  a server can miss the information that a stale value has been returned for some object by the transaction (which reads multiple objects), and return a fresh value for some other, violating causal consistency.

At the heart of our results lies essentially a fundamental trade-off between  causality and (eventual) freshness of values.\footnote{This trade-off is different from the traditional one in distributed computing between ensuring linearizability (i.e., finding a linearization point) and ensuring wait-freedom, both rather strong properties.} Understanding this trade-off is key to paving the path towards a new generation of transactional storage systems. Indeed,
the relevance of our results goes beyond the scope of causal consistency. They apply to any consistency model   stronger than causal consistency, e.g., linearizability \cite{herlihy_linearizability_1990, lamport_interprocess_1986}  and strict serializability  \cite{bernstein_concurrency_1987, papadimitriou_serializability_1979}, and are relevant also for systems that implement hybrid consistency models that include causal consistency, e.g., Gemini~\cite{osdi} and Indigo~\cite{other-from-balegas-esys}.

The rest of this paper is organized as follows.
Section 2 presents our general  model and definitions.
Section 3 presents the impossibility of fast read-only transactions.
Section 4 presents the impossibility of fast invisible read-only transactions (in the restricted model).
Section 5 discusses related work. 
Section 6 discusses how to circumvent our impossibility results. 
For space limitation, we defer the details of the proofs of our impossibility results  to the appendix.

\section{Model and Definitions}

\subsection{Model}
We assume an arbitrarily large number of \emph{clients} $C_1, C_2, C_3, \ldots$ (sometimes also denoted by $C$), and at least two \emph{servers} $P_X, P_Y$. Clients and servers interact by exchanging messages.
We consider an \emph{asynchronous} system where the delay on message transmission is finite but arbitrarily large, and there is no global clock accessible to any process.
Communication channels do not lose, modify, inject, or duplicate messages, but messages could be reordered.

A \emph{storage} is a finite set of objects. 
Clients read and/or write objects in the storage via \emph{transactions}. Any transaction $T$ consists of a read set $R_T$ and a write set $W_T$ on an arbitrary number of objects ($R_T$ or $W_T$ could be empty). We denote $T$ by $(R_T, W_T)$.
We say that a client \emph{starts} a transaction when the client \emph{requests} the transaction from the storage.
Any client which requests transaction $T$ returns a value for each read in $R_T$ and  $ok$ for each write in $W_T$. 
We say that a client \emph{ends} a transaction when the client returns from the transaction. Every transaction ends.

The storage is implemented by servers. For simplicity of presentation,
we assume that each server stores a different set of objects and the set is disjoint among servers.
(We show in the appendix how our results apply to the non-disjoint case.)
Every server receiving a request from a client responds. A server's response without any client request is not considered, and no server receives requests for objects not stored on that server. Naturally, a server that does not store an object stores no information on values written to that object. Clients do not buffer the value of an object to be read; instead a server returns one and only one value which has been written to the object in question. 

\subsection{Causality}

We consider a storage that ensures \emph{causality} in the classical sense of \cite{ahamad_causal_1995}, which we first recall and adapt to a transactional context.

The \emph{local history} of client $C_i$, denoted  $L_i$, is a sequence of start and end events.
We assume, w.l.o.g., that any client starts a new transaction after the client has ended all previous transactions, i.e., clients are sequential. Hence any local history $L_i$ can be viewed as a sequence of transactions.
We denote by $r(x)v$ a read on object $x$ which returns $v$, by $r(x)*$ a read on object $x$ for an unknown return value (with symbol $*$ as a place-holder), and by $w(x)v$ a write of $v$ to object $x$. For simplicity, we assume that every value written is unique. (Our results hold even when the same values can be written.) Definition \ref{def:ahamad_causal} captures the program-order and read-from causality relation \cite{ahamad_causal_1995}.

\begin{definition}[Causality \cite{ahamad_causal_1995}]
\label{def:ahamad_causal}
Given local histories $L_1, L_2, L_3\ldots$, for any $\alpha=a(x_\alpha)v_\alpha,  \beta=b(x_\beta)v_\beta$ where $a,b\in\{r,w\}$, we say that $\alpha$ causally precedes $\beta$, which we denote by $\alpha \rightsquigarrow \beta$, if
(1) $\exists i$ such that $\alpha$ is before $\beta$ in $L_i$; or
(2) $\exists v, x$ such that $\alpha = w(x)v$ and $\beta = r(x)v$; or
(3) $\exists \gamma$ such that $\alpha\rightsquigarrow \gamma$ and $\gamma\rightsquigarrow\beta$.
\end{definition}

Our definition of causally consistent transactions follows closely the original definition of \cite{ahamad_causal_1995}. We only slightly extend the classical definition of causal serialization in \cite{ahamad_causal_1995} to cover transactions. 
Assume that each object is initialized with a special symbol $\bot$. (Thus a read can be $r(x)\bot$.)

\begin{definition}[Transactional causal serialization]
\label{def:serialization}
Given local histories $H = L_1, L_2, L_3, \ldots$, we say that client $C_i$'s history can be causally serialized if we can totally order all transactions that contain a write in $H$ and all transactions in $L_i$, such that
(1) for any read $r = r(x)v$ on object $x$ which returns a non-$\bot$ value $v$, the last write $w(x)v_w$ on $x$ which precedes the transaction that contains $r$ satisfies $v_w = v$;
(2) for any read $r = r(x)\bot$ on object $x$, no write on $x$ precedes the transaction that contains $r$;
(3) for any $\alpha, \beta$ such that $\alpha\rightsquigarrow \beta$, the transaction that contains $\alpha$ is ordered before the transaction that contains $\beta$.
\end{definition}

\begin{definition}[Causally consistent transactional causal storage]
\label{def:causal-consistency}
We say that storage $cc$ is causally consistent if for any execution of clients with $cc$, each client's local history can be causally serialized.
\end{definition}

\subsection{Progress}

\emph{Progress} is necessary to make any storage useful; otherwise, an implementation which always returns $\bot$ or values written by the same client can trivially satisfy causal consistency. To ensure progress, we require any value written to be eventually  \emph{visible}. 
While rather weak, this definition is strong enough for our impossibility results, which  apply to stronger definitions.
We formally define progress in Definition \ref{def:visible} below. Existing implementations of causal consistency \cite{lloyd_settle_2011, almeida_chainreaction_2013, lloyd_stronger_2013, du_gentlerain_2014, akkoorath_cure_2016, mehdi_occult_2017, bravo_saturn_2017} indeed used the terminology of visible writes/updates/values and implicitly included progress as a property of their causally consistent systems, yet there has been no formal definition for progress.\footnote{Bailis et al. \cite{bailis_bolt-on_2013} defined eventual consistency in a similar way to progress here; however they considered progress only in the situation where all writes can stop.} 

\begin{definition}[Progress]
\label{def:visible}
A (causally consistent) storage guarantees progress if,  
for any write $w = w(x)v$, $v$ is eventually visible: there exists finite time $\tau_{x,v}$ such that any read $r(x)v_{new}$ which starts at time $t \geq \tau_{x,v}$, satisfies $v_{new} = v$ or $w(x)v_{new}$ returns no earlier than $w$ starts.\footnote{The accurate time is used for the ease of presentation for definitions and proofs and not accessible to any process.} 
\end{definition}

\section{The Impossibility of Fast Transactions} 
In this section, we present and prove our first theoretical result, Theorem \ref{thm:imp}. We first define formally the notion of \emph{fast} transactions.
In short, a fast transaction is one of which each  operation executes in (at most) one communication round between a client and a server (Definition \ref{def:fast-op} below). 

\begin{definition}[Fast transaction]
\label{def:fast-op}
We say that transaction $T$ is fast if for any client $C$ and $C$'s invocation $I$ of $T$, there is an execution where $I$ ends and during $I$, for any server $P$:
\begin{itemize}
\item $C$ sends at most one message to $P$ and receives at most one message from $P$;
\item If $C$ sends a message to $P$, then after the reception of that message,  any message which $P$ sends to a server is delayed and $P$ receives no message from any server until $I$ ends. 
\end{itemize}
\end{definition}

Definition \ref{def:fast-op} excludes implementations where 
a server  waits for the reception  of messages from another server (whether the server is one which $C$ sends a message to or not) to reply to a client.
Definition \ref{def:fast-op} allows parallel transactions.

\subsection{Result}
Theorem \ref{thm:imp} says that it is impossible to implement fast transactions (even if just read-only ones are fast).

\begin{theorem}
\label{thm:imp}
If a causally consistent transactional storage provides transactions that can read and/or write multiple objects, then no implementation provides fast read-only transactions.
\end{theorem}

The intuition behind Theorem \ref{thm:imp} is the following. Consider a server $P_X$ that stores object $X$ and a server $P_Y$ that stores object $Y$. If there is a risk of violating causality for $P_Y$ where  $P_X$  could return an old value, then $P_Y$ must also return an old value to the same transaction. In order to guarantee progress, extra communication is needed, which could further delay $P_Y$ from returning a new value, in turn, creating a risk of violating causality for $P_X$. In fact, $P_X$ and $P_Y$ could take turns creating causality violation risks for each other,  and preventing each other from returning new values forever, jeopardizing thereby progress.
For space limitation, we just sketch below our proof of Theorem \ref{thm:imp}. (The full proof  is deferred to the appendix.)

\subsection{Proof overview}

The proof of Theorem \ref{thm:imp} is by construction of a contradictory  execution $E_{imp}$ which, to satisfy causality, contains an infinite number of messages the reception of which is necessary for some value to be visible (violating progress).
As illustrated in Figure \ref{fig:eimp}, some non-$\bot$ values of $X$ and $Y$ have been visible in $E_{imp}$; then client $C_w$ issues transaction $WOT = (w(X)x, w(Y)y)$ which starts at  time $t_w$; since $t_w$, $WOT$ is the only executing transaction. We make no assumption on the execution of $WOT$.

We show an infinite number of messages by induction on the number $k$ of messages: no matter how many $k$ messages have been sent and received, an additional message is necessary for $x$ and $y$ to be visible. 
Let $m_0, m_1, \ldots, m_{k-1}, m_k$ be the sequence of messages for case $k$. We show that in $E_{imp}$, except for $m_0$ and $m_1$, every message is sent after the previous message has been received. At the end of the induction, we conclude that in $E_{imp}$, these messages delay both $x$ and $y$ from being visible. As every message is sent after previous messages have been received, the delay accumulates and thus violates Definition \ref{def:visible}.
We sketch below the proof of the base case and the inductive step.

\begin{figure}[!h]
    \centering
  \begin{subfigure}[b]{0.47\textwidth}
    \includegraphics[width=\textwidth]{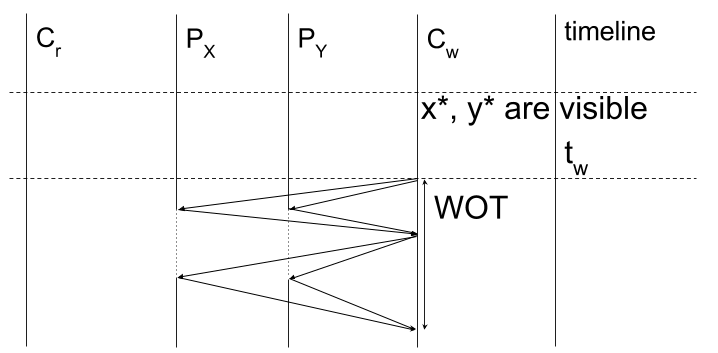}
    \caption{Construction of $E_{imp}$}
    \label{fig:eimp}
    \end{subfigure}
    ~~~~
  \begin{subfigure}[b]{0.47\textwidth}
    \includegraphics[width=\textwidth]{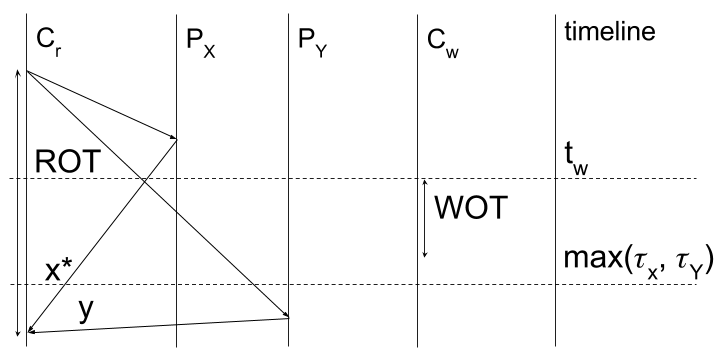}
    \caption{Contradictory execution for the existence of the first two messages}
    \label{fig:m0m1}
    \end{subfigure}
    \caption{Illustration of $E_{imp}$ and the base case}
    \label{fig:no-name}
\end{figure}

\subsection{Base case}
We first define some terminology to unify the description of communication between $P_X$ and $P_Y$ no matter whether the communication is via some third server or not: we say that $P_X$ ($P_Y$) sends a message which \emph{precedes} some message that arrives at $P_Y$ ($P_X$), in the sense defined below (Definition \ref{def:order}). Thus the case where $P_X$ sends message $m$ to server $S$ and $S$ forwards $m$ to $P_Y$ is covered.

\begin{definition}
\label{def:order}
Message $m_1$ \emph{precedes} message $m_2$ if (1) $m_1 = m_2$, or (2) a process sends $m_2$ after it receives $m_1$ or (3) there exists message $m$ such that $m_1$ precedes $m$ and $m$ precedes $m_2$.
\end{definition}

In the base case where $k = 1$, we show that after $t_w$, each server sends a message that precedes some message which arrives at the other server, by contradiction. By symmetry, suppose that $P_X$ sends no message that precedes any message which arrives at $P_Y$. Then we add a read-only transaction $ROT$ to $E_{imp}$,  illustrated in Figure \ref{fig:m0m1}: to $P_X$, the request of $ROT$ is earlier than that of $WOT$, while to $P_Y$, the request of $ROT$ is (much) later and is actually after $x$ and $y$ are eventually visible. By fast read-only transactions and our assumption for contradiction, after $t_w$, there can be no communication between $P_X$ and $P_Y$  before $P_Y$'s response. As a result, $ROT$ returns $(x^*, y)$ for some $x^*\neq x$. Lemma \ref{lma:atomicity} (of which the proof is also deferred to the appendix) depicts the very fact that such returned value violates causal consistency. A contradiction.
By symmetry, we conclude that both $P_X$ and $P_Y$ have to send at least one message after $t_w$. These two messages are $m_0$ and $m_1$. Let $\{P,Q\} = \{P_X, P_Y\}$. Clearly, one server $P$ between $P_X$ and $P_Y$ sends its message earlier than the other server $Q$. We let $m_0$ be the message sent by $P$ and $m_1$, the other message. 

\begin{lemma}
\label{lma:atomicity}
In $E_{imp}$, no write (including writes in a transaction) occurs other than $WOT$ since $t_w$. If some client $C_r$ requests $ROT$,
then $ROT$ returns $x$ if and only if $ROT$ returns $y$.
\end{lemma}

\subsection{Inductive step}

From case $k = 1$ to case $k=2$, we show that another message $m_2$ is necessary (for the value written at $Q$ to be visible). Let $m$ be the first message which $P$ receives and $m_1$ precedes. We argue by contradiction. Suppose that after the reception of $m$, $P$ sends no message that precedes any message which arrives at $Q$. If the request of $ROT$ comes at $P$ after $P$ sends $m_0$ and before $P$ receives $m$,  then by Lemma \ref{lma:atomicity}, $P$ must return some $x^*\neq x$ or some $y^*\neq y$, considering the possibility that the request of $ROT$ could come at $Q$ before $t_w$, illustrated in Figure \ref{fig:oldvalue}. 
Now the request of $ROT$ actually comes (much) later at $Q$ (after the value written at $Q$ is visible), illustrated in Figure \ref{fig:newmsg}.
By fast read-only transactions and our assumption for contradiction, after $Q$ receives all messages preceded by $m_0$, there can be no communication between $P$ and $Q$ before $Q$'s response. As a result, $Q$ returns $x$ or $y$ and $ROT$ returns $(x^*,y)$ or $(x,y^*)$, violating Lemma \ref{lma:atomicity}. A contradiction. 

\begin{figure}[!h]
    \centering
  \begin{subfigure}[b]{0.37\textwidth}
    \includegraphics[width=\textwidth]{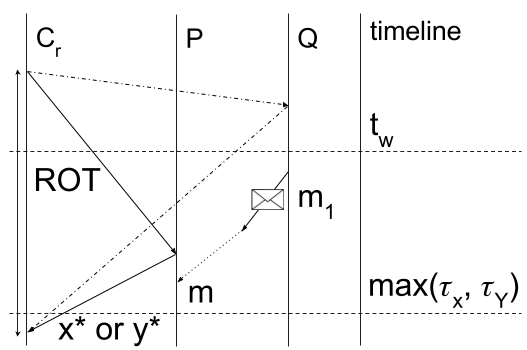}
    \caption{Return of old values}
    \label{fig:oldvalue}
    \end{subfigure}
    ~~~~
  \begin{subfigure}[b]{0.37\textwidth}
    \includegraphics[width=\textwidth]{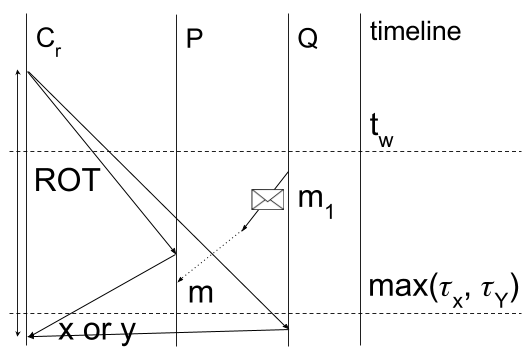}
    \caption{Contradictory execution}
    \label{fig:newmsg}
    \end{subfigure}
    \caption{Existence of more messages after $WOT$}
    \label{fig:no-name-3}
\end{figure}

We then conclude that $P$ must send $m_2$ after the reception of $m$, which is no earlier than the reception of $m_1$. For case $k = 3$, we can similarly show that $Q$ must send another message $m_3$ (for the value written at $P$ to be visible).
In this way, we add one message in each step of the induction, while $P_X$ and $P_Y$ take turns in sending messages necessary for $x$ and $y$ to be visible.
As shown by induction, the total number of messages essentially grows to infinity. This completes our construction of $E_{imp}$ as well as our proof sketch of Theorem \ref{thm:imp}.

\section{The Impossibility of Fast Invisible Transactions}
As we pointed out in the introduction, some systems considered a restricted model where all transactions are read-only and write operations are supported only outside the scope of a transaction.
This restricted model also circumvents the impossibility result of Theorem \ref{thm:imp}.
In this model, we present our second theoretical result, Theorem \ref{thm:main}, stating that fast read-only transactions (while indeed possible) need to be visible (need to actually write). We first formally define the notion of (in)visible transactions in Definition \ref{def:trace} below.

\begin{definition}[Invisible transactions]
\label{def:trace}
We say that transaction $T$ is invisible if for any client $C$ and $C$'s invocation $I$ of $T$, any execution $E$ (until $I$) can be continued arbitrarily but still there exists some execution $E^-$ without $I$ that is the same as $E$ except for the message exchange with $C$ (during the time period of $I$). 
\end{definition}

\subsection{Result}
Theorem \ref{thm:main} shows that it is impossible to implement fast invisible transactions (even if all transactions are read-only). 

\begin{theorem}
\label{thm:main}
If a causally consistent transactional storage provides fast read-only transactions, then no implementation provides invisible read-only transactions.
\end{theorem}

The intuition of Theorem \ref{thm:main} is the following. In an asynchronous system, any read-only transaction $T$ can read an old value and a new value from different servers, and thus the communication that carries $T$ is necessary to prevent $T$ from returning a mix of old and new values.
For space limitation, below we sketch our proof of Theorem \ref{thm:main}. (The full proof is deferred to the appendix.)

\subsection{Crucial executions}
To prove Theorem \ref{thm:main}, 
we consider any execution $E_1$ where some client $C_r$ (which has not requested any operation before) starts transaction $ROT = (r(X)*, r(Y)*)$ at the  time $t_0$. In $E_1$, before $t_0$, some values of $X$ and $Y$ have been visible. We continue $E_1$ with some client $C$ executing $w(X)x$ and $w(Y)y$ (which establishes $w(X)x\rightsquigarrow w(Y)y$).

Our proof is by contradiction. Suppose that transaction $ROT$ is invisible. Then no matter how $E_1$ is scheduled, there exists some execution $E_2$ such that $E_2$ is the same as $E_1$ except that (1) $C_r$ does not invoke $ROT$, and (2) the message exchange with $C_r$ during the time period of $ROT$ is different. Below we first schedule $E_1$ and then construct another execution $E_{1,2}$. We later show $E_{1,2}$ violates causal consistency.
By fast read-only transactions, we can schedule messages such that the message which $C_1$ sends during $ROT$ arrives at $P_X$ and $P_Y$ respectively at the same time. Let $T_1$ denote this time instant and let $T_2$ be the time when $ROT$ eventually ends, illustrated in Figure \ref{fig:edi-new}. During $[T_1, T_2]$, $P_X$ and $P_Y$ receive no message but still respond to $C_r$.
After $T_2$, the two writes of $C$ occur, while $C_r$ does no operation.
All delayed messages eventually arrive before $y$ can be visible. In $E_1$, $y$ is visible after some time $\tau_y$.

\begin{figure}[!h]
    \centering
  \begin{subfigure}[b]{0.44\textwidth}
    \includegraphics[width=\textwidth]{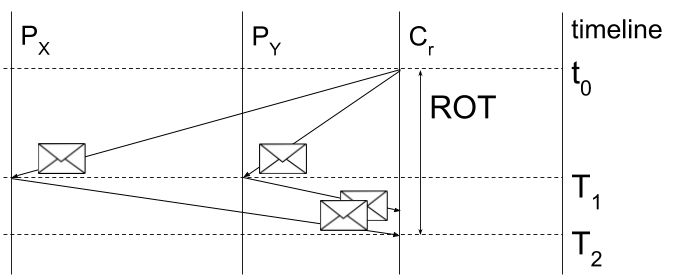}
    \caption{Message schedule of $E_1$}
    \label{fig:edi-new}
    \end{subfigure}
  \begin{subfigure}[b]{0.44\textwidth}
    \centering
    \includegraphics[width=\textwidth]{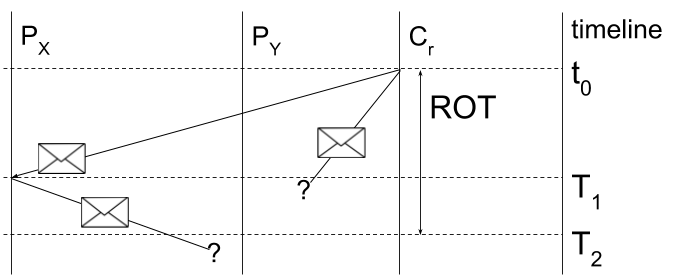}
    \caption{Message schedule of $E_{1,2}$}
    \label{fig:ed1d2-new}
    \end{subfigure}
    \caption{Construction and extension of $E_1$}
    \label{fig:no-name-2-new}
\end{figure}

Next we construct execution $E_{1,2}$ that is indistinguishable from $E_1$ to $P_X$ and from $E_2$ to $P_Y$. The start of $E_{1,2}$ is the same as $E_1$ (as well as $E_2$) until $t_0$.
At $t_0$, $C_r$ still invokes $ROT$.
As illustrated in Figure \ref{fig:ed1d2-new}, $P_X$ receives the same message from $C_r$ and sends the same message to $C_r$ at the same time as in $E_1$;
$C_r$ sends the same message to $P_Y$ at the same time as in $E_1$, the reception of which is however delayed by a finite but unbounded amount of time. In addition, during $[T_1, T_2]$, $P_X$ and $P_Y$ receive no message as in $E_1$ (as well as $E_2$).
Thus by $T_2$, $P_X$ is unable to distinguish between $E_1$ and $E_{1,2}$ while $P_Y$ is unable to distinguish between $E_2$ and $E_{1,2}$. 
According to our assumption for contradiction, $E_{1,2} = E_1 = E_2$ except for the communication with $C_r$ by $T_2$.

\subsection{Proof overview}

We continue our proof of Theorem \ref{thm:main} (by contradiction). Based on the executions constructed above, we extend $E_2$ and $E_{1,2}$ after $\tau_y$. As illustrated in Figure \ref{fig:extend-new}, we let $C_r$ start $ROT$ immediately after $\tau_y$ in $E_2$. 
In both $E_2$ and $E_{1,2}$, by fast read-only transactions, we schedule the message sent from $C_r$ to $P_Y$ during $C_r$'s $ROT$ to arrive at the same time after $\tau_y$, and $\exists t$ such that during $[\tau_y, t]$, $P_Y$ receives no message but still responds to $C_r$. 
By $t$, $P_Y$ is unable to distinguish between $E_2$ and $E_{1, 2}$.

\begin{figure}[!h]
    \centering
  \begin{subfigure}[b]{0.4\textwidth}
    \includegraphics[width=\textwidth]{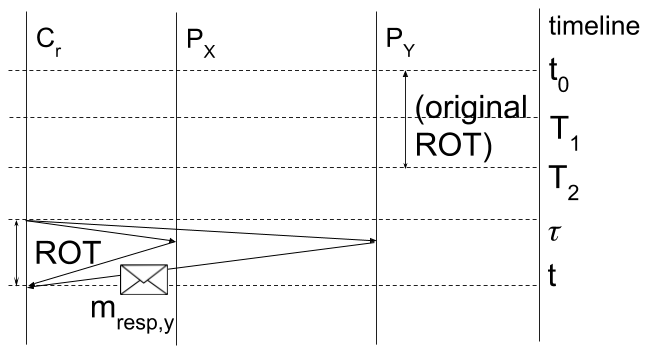}
    \caption{Extension of $E_2$}
    \label{fig:extend-e2-new}
    \end{subfigure}~~~~~~~~
  \begin{subfigure}[b]{0.36\textwidth}
    \includegraphics[width=\textwidth]{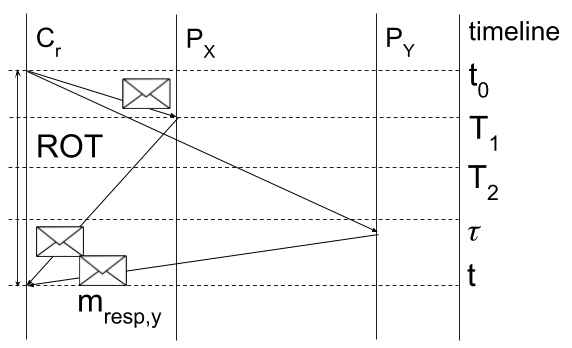}
    \caption{Extension of  $E_{1,2}$.}
    \label{fig:extend-ed1d2-new}
    \end{subfigure}
    \caption{Extension of crucial executions}
    \label{fig:extend-new}
\end{figure}

We now compute the return value of $ROT$ in $E_{1,2}$. By progress, in $E_2$, $P_Y$ returns $y$, and then by indistinguishability, in $E_{1,2}$, $P_Y$ also returns $y$.
Since in $E_{1,2}$, $P_X$ returns some value $x^*\neq x$ (as $w(X)x$ starts after $T_2$), the return value of $ROT$ in $E_{1,2}$ is $(x^*, y)$.
According to our assumption, $E_{1,2}$ satisfies causal consistency.
By Definition \ref{def:causal-consistency}, 
we can totally order all $C_r$'s operations and all write operations in $E_{1,2}$ such that the last preceding writes of  $X$ and $Y$ before $C_r$'s $ROT$ are $w(X)x^*$ and $w(Y)y$ respectively. 
This leads $w(X)x^*$ to be ordered after $w(X)x$. However, if we extend $E_{1,2}$ so that $C_r$ invokes $ROT_1 = (r(X)*, r(Y)*)$ after $x$ and $y$ are visible, then $ROT_1$ returns value $(x, y)$ and if we do total ordering of $E_{1,2}$ again, then the last preceding write of $X$ before $ROT_1$ must be $w(X)x$, contradictory to the ordering between $w(X)x^*$ and $w(X)x$.
Therefore, we conclude that $E_{1,2}$ violates causal consistency, which completes our proof sketch of Theorem \ref{thm:main}.

\section{Related Work}

\subsection{Causal consistency}
Ahamad et al. \cite{ahamad_causal_1995} were the first to propose \emph{causal consistency} for a memory accessed by read/write operations.
Bouajjani et al. \cite{bouajjani_verifying_2017} formalized the verification of causal consistency.
A large number of systems \cite{lloyd_settle_2011, almeida_chainreaction_2013, du_orbe_2013, lloyd_stronger_2013, du_gentlerain_2014, mehdi_occult_2017} implemented transactional causal consistency, although none formalized the concept for generic transactions.
Akkoorath et al. \cite{akkoorath_cure_2016} extended causal consistency to transactions by defining atomicity for writes and causally consistent snapshots for reads within the same transaction. 
Mehdi et al. \cite{mehdi_occult_2017} introduced \emph{observable causal consistency} in the sense that each client observes a monotonically non-decreasing set of writes. Neither of the two definitions follows a formalization close to the original definition of \cite{ahamad_causal_1995}.

\subsection{Causal read-only transactions}

Most implementations do not provide fast (read-only) transactions.
COPS \cite{lloyd_settle_2011} and Eiger \cite{lloyd_stronger_2013} provide a two-round protocol for read-only transactions.
Read-only transactions in Orbe \cite{du_orbe_2013}, GentleRain \cite{du_gentlerain_2014}, Cure  \cite{akkoorath_cure_2016} and Occult \cite{mehdi_occult_2017} can induce more than one-round communication.
Read-only transactions in ChainReaction \cite{almeida_chainreaction_2013} can induce more than one-round communication as well as abort and retry, resulting in more communication.
Eiger-PS \cite{lu_snow_2016} provides fast transactions and satisfies \emph{process-ordered serializability} \cite{lu_snow_2016}, stronger than causal consistency; yet in addition to the request-response of a transaction, each client periodically communicates with every server. Our Theorem \ref{thm:imp} explains Eiger-PS's additional communication. 
COPS-SNOW \cite{lu_snow_2016} provides fast read-only transactions but writes can only be performed outside a transaction; moreover, any read-only transaction in COPS-SNOW is visible, complying with our Theorem \ref{thm:imp} and Theorem \ref{thm:main}.
If each server stores a copy of all objects, then SwiftCloud \cite{zawirski_write_2015} provides fast read-only transactions. However, it is not clear how such a storage can  scale well with the growth of data given a single server. SwiftCloud considers the storage of a full copy among multiple servers and its resulting parallelism is an orthogonal issue \cite{zawirski_write_2015}.

\subsection{Impossibility results}

Existing impossibility results on storage systems have typically considered  stronger consistency properties than causality  or stronger progress conditions than eventual visibility. 
Brewer \cite{brewer_conjecture_2000} conjectured the CAP theorem that no implementation  guarantees \emph{consistency}, and  \emph{availability} despite \emph{partitions}.
Gilbert and Lynch \cite{gilbert_brewer_2002} formalized and proved Brewer's conjecture in \emph{partially synchronous} systems. They formalized consistency by \emph{atomic} objects \cite{lamport_interprocess_1986} (which satisfy \emph{linearizability} \cite{herlihy_linearizability_1990}, stronger than causal consistency). 
Considering a storage implemented by \emph{data centers} (clusters of servers), if any value written is \emph{immediately} visible to the reads at the same data center (as the write), and a client can access different objects at different data centers, Roohitavaf et al. \cite{roohitavaf_causalspartan_2017} proved the impossibility of ensuring causal consistency and  availability despite partitions. Their proof (as well as the proof of the CAP Theorem) rely on message losses.
Lu et al. \cite{lu_snow_2016} proved the SNOW theorem, saying that fast \emph{strict serializable}   transactions \cite{bernstein_concurrency_1987, papadimitriou_serializability_1979} (satisfying stronger consistency than causal consistency) are impossible. Their proof assume writes that are also fast, and does not imply our proof of impossibility results.

Mahajan et al. \cite{mahajan_consistency_2011}, Attiya et al. \cite{attiya_limitations_2017} as well as    Xiang and Vaidya \cite{xiang_lower_2017} proposed  related notions of causal consistency  based on the events at servers (rather than clients) motivated by replication schemes (an issue orthogonal to the problem considered in this paper).
More specifically, Mahajan et al. \cite{mahajan_consistency_2011} proved the CAC theorem that no implementation guarantees \emph{one-way convergence},\footnote{A progress condition based on the communication between servers.} availability, and any consistency stronger than \emph{real time causal consistency} assuming infinite local clock events and arbitrary message loss. 
Attiya et al. \cite{attiya_limitations_2017} proved that a \emph{replicated} store implementing multi-valued registers cannot satisfy any consistency strictly stronger than \emph{observable causal consistency}.\footnote{The definitions of observable causal consistency given by Mehdi et al. and Attiya et al.   \cite{attiya_limitations_2017, mehdi_occult_2017} are different.} 
Xiang and Vaidya \cite{xiang_lower_2017} proved that for \emph{replica-centric causal consistency}, it is necessary to track down writes.

\subsection{Transactional memory}
In the context of transactional memory, if the implementation of a read-only operation (in a transaction) writes a base shared object, then the read-only operation is said to be \emph{visible} and \emph{invisible} otherwise \cite{guerraoui_correctness_2008}.
Known impossibility results on invisible reads of TM assume stronger consistency than causal consistency.
Attiya et al. \cite{attiya_inherent_2011} showed that no TM implementation ensures strict serializability, \emph{disjoint-access parallelism} \cite{attiya_inherent_2011}\footnote{Disjoint-access parallelism \cite{attiya_inherent_2011} requires  two transactions  accessing different application objects to also access different base objects.} and uses invisible reads, the proof of which shows that if writes are frequent, then a read can miss some write forever.
Peluso et al. \cite{peluso_disjoint_2015} considered any consistency that respects the real-time order of transactions (which causal consistency does not necessarily respect), and proved a similar impossibility result. 
Perelman et al. \cite{perelman_maintaining_2010} proved an impossibility result for a multi-version TM implementation with invisible read-only transactions that ensures strict serializability and maintains only a necessary number of versions, the proof of which focuses on garbage collection of versions. None of the results or proofs above imply our impossibility results.

\section{Concluding Remarks}

Our impossibility results establish fundamental limitations on the performance on transactional storage systems. The first impossibility basically says that fast read-only transactions are impossible in a general setting where writes can also be performed within transactions.  The second impossibility says that in a setting where all transactions are read-only, they can be fast, but they need to visible. A system like COPS-SNOW \cite{lu_snow_2016} implements such visible read-only transactions that leave traces when they execute, and these traces are propagated on the servers during writes. (For completeness, we sketch in Appendix D a variant algorithm where these traces are propagated asynchronously, i.e., outside writes). 

Clearly, our impossibilities apply to causal consistency and hence  to any stronger consistency criteria. They hold without assuming any message or node failures and hence hold for failure-prone systems. For presentation simplicity, we assumed that servers store disjoint sets of objects, but our impossibility results hold without this assumption (Appendix C). 
Some design choices could circumvent these impossibilities like imposing a full copy of all objects on each server (as in SwiftCloud \cite{zawirski_write_2015}), periodic communication between servers and clients (as in Eiger-PS \cite{lu_snow_2016}), 
or transactions that abort and retry (as in ChainReaction \cite{almeida_chainreaction_2013}),
\footnote{Multiple versions (allowed to be returned in a transaction) do not circumvent our impossibility results as an infinite number of versions would be necessary.}
Each of these choices clearly hampers scalability. 

We considered an asynchronous system where messages can be delayed arbitrarily   and there is no global clock. One might also ask what happens with synchrony assumptions.
If we assume a fully synchronous system where message delays are bounded  and all processes can access a global accurate clock, then our impossibility results can be both circumvented. We give such a timestamp-based algorithm in Appendix D.
If we consider however a system where communication delays are  unbounded and all processes can access a global clock, then only our Theorem \ref{thm:imp} holds. In this sense, message delay is key to the impossibility of fast read-only transactions, but not to the requirement that they need to be visible, in the restricted model where all transactions are read-only. In this restricted model, our timestamp-based algorithm of Appendix D can also circumvent Theorem \ref{thm:main} if we assume a global clock.

\newpage

\bibliographystyle{IEEEtran}
\bibliography{IEEEabrv,bib}

\newpage

\begin{appendix}

\section{Proof of Theorem \ref{thm:imp}}

\subsection{Definition of One-Version Property}

In Section 2, we required that a server returns one and only one value which has been written to an object, a property we we call one-version, which we define below.   The formal definition is necessary because (1) there are a lot of possibilities for message $m$ to return value $v$, e.g., $m = v$, or $m = v \mbox{ XOR } c$, or $m = v + c$ for some constant $c$; and (2) if messages $m_1$ and $m_2$ are from two different servers $P_X$ and $P_Y$ and $m_1 = (x, \mbox{first 8 bits of }z \mbox{ XOR } c)$,  $m_2 = (y, \mbox{other bits of } z \mbox{ XOR } c)$, where $z$ is a value written to another object $Z$, then $(m_1, m_2)$ can return more values $x, y, z$ than expected. The first issue calls for defining messages in a general manner; the second situation should be excluded (as it is not implemented by any practical storage system to the best of our knowledge). 
The formal definition addresses both issues.

To define the value included in a message in general, we have to measure the information \emph{revealed} by events and messages.
We consider the maximum amount of information that any algorithm can output according to the given input: events and messages. We then restrict the class of algorithms any correct implementation may provide (to the client-side). For example, an algorithm that outputs 1 regardless of the input should be excluded. Hence two definitions, one on algorithms used to reveal information and one on information indeed revealed, are presented before the definition of one-version property.

\begin{definition}[Successful algorithms]
Consider any algorithm, denoted by $\mathcal{A}$, whose input is some information $i_E$ (events and messages) of execution $E$. The output of $\mathcal{A}$ is denoted by $\mathcal{A}(i_E)$. We say that $\mathcal{A}$ is \emph{successful}
\begin{itemize} 
\item If $v\in \mathcal{A}(i_{E^v})$, then in $E^v$, $w(a)v$ occurs; and 
\item For any value $u$, let $E^u$ be the resulting execution where $w(a)v$ is replaced by $w(a) u$. Then $u\in\mathcal{A}(i_{E^u})$.
\end{itemize} 
\end{definition}

\begin{definition}[Information revealed]
Consider execution $E$, client $C$ and $C$'s invocation $I$ of some transaction. Denote by $M$ any non-empty subset of message receiving events that occur at $C$ (including message contents) during $I$. 
We say that $M$ \emph{reveals} $(n_2 - n_1)$ version(s) of an object $a$ if
\begin{itemize}
\item Among all successful algorithms whose input is $v_{C, I}$, $n_1$ is the maximum number of values in the output that are also values written to $a$ before the start of $I$;
\item Among all successful algorithms whose input is $v_{C, I}$ and $M$, $n_2$ is the maximum number of values in the output that are also values written to $a$ before the end of $I$;
\end{itemize} 
where $v_{C, I}$ is $C$'s \emph{view}, or all events that have occurred at $C$ (including the message content if an event is message receiving), before the start of $I$.
\end{definition}

\begin{definition}[One-version property]
\label{def:one-version}
Consider any execution $E$, any client $C$ and $C$'s invocation $I$ of an arbitrary transaction $T$ with non-empty read set $R$. ($T$ is general here in that $T$ may contain only a single read, i.e., the write set is empty and $|R| = 1$). For any non-empty set of servers $A$, let $\Lambda_{I, A} = R\cap\{ \mbox{objects stored on } P |\forall P\in A\}$ and denote by $M_{I, A}$ the events of $C$ receiving messages from any server in $A$ (including message contents) during $I$. Then an implementation satisfies one-version property if
\begin{itemize}
\item $\forall E, \forall I, \forall A$, $M_{I, A}$ reveals at most one version for each object in  $\Lambda_{I, A}$, and no version of any object not in $\Lambda_{I, A}$; and
\item $\forall E, \forall I$, when $A$ includes all servers, then $M_{I, A}$ reveals exactly one version for each object in  $R$, and no version of any object not in $R$.
\end{itemize}
(If $M_{I, A}$ reveals exactly one version of an object $a$, we may also specify the version $v$ and say that $M_{I, A}$ reveals $v$.)
\end{definition}

A final remark is on the relation with the property of fast transactions. Naturally, when we consider the maximum amount of information revealed by a transaction, we have to consider all message receiving events at the client-side. As one-version property is defined in general here (independent from the property of fast transactions), there can be multiple message receiving events during a transaction. The formal definitions above consider the set of all these events rather than individual ones separately (i.e., what one message can reveal). 
This general definition is necessary to disallow implementations equivalent to fast transactions to bypass our results. Consider an equivalent implementation in a transaction of which a server splits its message to several ones and sends them to a client where each message reveals one version. Such implementation does not conform to our requirement on servers in Section 2 yet is however not excluded by a definition considering only individual messages, showing the necessity of a general definition as we present above.

\subsection{Construction of $E_{imp}$}

The construction of $E_{imp}$ is based on the following notations and execution $E_{prefix}$.
We denote by $P_X$ the server which stores object $X$, and $P_Y$ the server which stores object $Y$. Let $E_{prefix}$ be any execution where $X$ and $Y$ have been written at least once and some values of $X$ and $Y$ have been visible. Denote by $t_{start}$ when some values of $X$ and $Y$ have been visible in $E_{prefix}$. 
Then we construct execution $E_{imp}$ starting from $t_{start}$. In $E_{imp}$, client $C_w$ does transaction $WOT = (w(X)x, w(Y)y)$ which starts at some time $t_w > t_{start}$, while all other clients do no transaction.
For $E_{imp}$, since $t_w$, $WOT$ is the only transaction. However, for \emph{any} positive number $k$, we show that $k$ messages have to be sent and received after $t_w$ and before $x$ and $y$ are visible. Since $k$ can be any positive number, then $k$ essentially goes to infinity. 

More specifically, we show that no matter how many $k$ messages have been sent and received, an additional message is necessary for $x$ and $y$ to be visible. I.e., our construction is by mathematical induction on the number $k$ of messages, summarized in Proposition \ref{prop:k}, Proposition \ref{prop:k-2} and Proposition \ref{prop:base}. 
Each case $k$ (Proposition \ref{prop:k-2}) is a property of $E_{imp}$ after $k$ messages have been sent and received in the visibility of values written by $WOT$: if one read-only transaction $ROT$ is added, then the values written by $WOT$ cannot yet be returned to $ROT$. 
Here are some notations for $ROT$ and the message schedule during $ROT$ which we use in the statement of case $k$. Let $C_r$ be the client which requests  $ROT = (r(X)*, r(Y)*)$; $C_r$ has requested no transaction before. By Definition \ref{def:fast-op}, for any $ROT$, we schedule messages such that every message which $C_r$ sends to either $P\in\{P_X,P_Y\}$ during $ROT$ arrives at the same time $t_P$ at $P$. After $t_P$ and before $P$ has sent one message to $C_r$ (during $ROT$), $P$ receives no message and any message sent by $P$ to a process other than $C_r$ is delayed to arrive after $ROT$ ends. For either $P$, we denote these messages which $P$ sends to $C_r$ after $t_P$ (during $ROT$) by $m_{resp,P}$. The message schedule of $ROT$ such that $P$ receives no message during $ROT$ is also assumed in Lemma \ref{lma:helper}.
Each case $k > 1$ is accompanied by a preliminary (Proposition \ref{prop:k}) on the necessity of an additional $k$th message, while the base case is a special case for which two additional messages are necessary (Proposition \ref{prop:base}). 
To deal with the special case, we index these messages starting from $0$: $m_0, m_1, \ldots, m_{k-1}, m_k$ (but our base case is still the case where $k = 1$). 
As shown in Proposition \ref{prop:base}, $P_X$ and $P_Y$ send $m_X$ and $m_Y$ after $t_w$ that precede some message which arrive at $P_Y$ and $P_X$ respectively. We define $m_0$ and $m_1$ as follows so that the base case is the case where $k = 1$ defined in Proposition \ref{prop:k}: one server between $P_X$ and $P_Y$ sends $m_{0}, m_{0}\in\{m_X, m_Y\}$ before receiving any message which is preceded by $m_1$ for $\{m_0, m_1\} = \{m_X, m_Y\}$. 
We refer to Definition \ref{def:order} for the formal definition on the relation of one message \emph{preceding} another used in our Proposition \ref{prop:k}, Proposition \ref{prop:k-2} and Proposition \ref{prop:base}.

\begin{proposition}[Additional message in case $k$]
\label{prop:k}
In $E_{imp}$, $m_0, m_1, \ldots, m_{k-1}$ have been sent.
Let $D_{k-1}$ be the source of $m_{k-1}$. Let $\{D_{k-1}, D_k\} = \{P_X, P_Y\}$. Let $T_{k-1}$ be the time when the first message preceded by $m_{k-1}$ arrives at $D_{k}$.
After $T_{k-1}$, $D_k$ must send at least one message $m_{k}$ that precedes some message which arrives at $D_{k-1}$.
\end{proposition}

\begin{proposition}[Case $k$]
\label{prop:k-2}
In $E_{imp}$, $m_0, m_1, \ldots, m_{k-1}, m_k$ have been sent. Then for any $t$ before $T_k$, if $C_r$ starts $ROT$ before $t$ and $t_{D_{k-1}} = t$, then $ROT$ may not return $x$ or $y$. 
\end{proposition}

\begin{proposition}[Additional message in the base case]
\label{prop:base}
After $t_w$, any $P\in\{P_X, P_Y\}$ must send at least one message that precedes some message which arrives at $Q$ for $\{P,Q\} = \{P_X, P_Y\}$.
\end{proposition}

\subsection{Proof of Theorem \ref{thm:imp}}

Before the proof of the base case and the inductive step from case $k$ to case $k+1$, we prove a helper lemma, Lemma \ref{lma:helper}. Lemma \ref{lma:helper} is helpful for the proof of both the base case and case $k$, and thus proved additionally to avoid repetition.
We also show a property of write-only transactions in Lemma \ref{lma:atomicity}. We refer to the main paper for its formal statement. As Lemma \ref{lma:helper} is based on Lemma \ref{lma:atomicity}, we prove the latter first. 

\begin{proof}[Proof of Lemma \ref{lma:atomicity}]
By contradiction. Suppose that for some execution $E_{imp}$ and some read-only transaction $ROT$, $ROT$ returns $(x^*, y)$ for some $x^*\neq x$, or $(x, y^*)$ for some $y^*\neq y$. By symmetry, we need only to prove the former. As $ROT$ returns $(x^*, y)$, by causal consistency, for $C_r$, there is serialization $\mathcal{S}$ that orders all $C_r$'s transactions and all transactions including a write such that the last preceding writes of $X$ and $Y$ before $ROT$ in $\mathcal{S}$ are $w(X)x^*$ and $w(Y)y$ respectively. Therefore $\mathcal{S}$ must order $WOT$ before $w(X)x^*$. However, if we extend $E_{imp}$ with $C_r$ requesting another read-only transaction $ROT_2$, then by progress, some $ROT_2$ must return $(x,y)$. As $ROT_2$ occurs after $ROT$, $\mathcal{S}$ must order $ROT_2$ after $ROT$ and then the last preceding writes of $X$ and $Y$ before $ROT_2$ in $\mathcal{S}$ cannot be $w(X)x$ and $w(Y)y$ respectively, contradictory to the property of causal consistency.
\end{proof}

\begin{lemma}[Communication prevents latest values]
\label{lma:helper}
Suppose that $E_{imp}$ has been extended to some time $A$ and there is no other write than contained in $WOT$ since $t_{start}$. Let $\{P, Q\} = \{P_X, P_Y\}$ where $P$ can be either $P_X$ or $P_Y$. Given $P$, assume that for some time $B > A$ and any $t_P\in [A,B)$, if $C_r$ starts $ROT$ before $t_P$, then $ROT$ may not return $x$ or $y$.
We have:
\begin{enumerate}
\item 
After $B$, $P$ must send at least one message which precedes some message that arrives at $Q$;
\item Let $t$ be the time when $Q$ receives the first message which is preceded by some message which $P$ sends after $B$. For any $\tau\in[A, t)$, if $C_r$ starts $ROT$ before $\tau$ and $t_Q = \tau$,\footnote{If needed, by the asynchronous communication, we may delay $t$ after $ROT$ ends to respect the message schedule of $ROT$ that $Q$ receives no message during $ROT$.} then $ROT$ may not return $x$ or $y$.
\end{enumerate}
\end{lemma}

\begin{proof}[Proof of Lemma \ref{lma:helper}]
We prove the first statement by contradiction. Suppose that after $B$, $P$ sends no message that precedes any message that arrives at $Q$. 
Let $t_s$ be the latest time before $B$ such that $P$ sends a message that precedes some message which arrives at $Q$ in $E_{imp}$. 
After $t_s$, we extend $E_{imp}$ into two different executions $E_1$ and $E_2$. 
Execution $E_2$ is $E_{imp}$ extended without any transaction. Thus $x$ and $y$ are visible after some time $t_{ev}$. Based on our assumption, $t_{ev} \geq B$. In $E_2$, $C_r$ starts $ROT$ after $t_{ev}$.
In $E_1$, $C_r$ starts $ROT$ after $t_s$ (and before $B$) and some $t_P\in[A, B)$. We delay any message which $P$ sends after $t_s$ in $E_1$.
Furthermore, in both $E_1$ and $E_2$, $t_Q > t_{ev}$. 
According to our assumption, after $t_s$, $P$ does not send any message which precedes some message that arrives at $Q$ in $E_2$. As we delay the messages which $P$ sends after $t_s$ in $E_1$, thus before $t_Q$, $Q$ is unable to distinguish between $E_1$ and $E_2$. After $t_Q$ (inclusive), according to the message schedule of $ROT$, by the time when $Q$ has sent one message to $C_r$ during $ROT$, $Q$ is still unable to distinguish between $E_1$ and $E_2$.
In $E_2$, since $C_r$ starts $ROT$ after $t_{ev}$, $ROT$ returns $(x,y)$ by progress.
The client-side algorithm $\mathcal{A}$ of $C_r$ to output the return value of $ROT$ is a successful algorithm. Since given $m_{resp, P}$ and $m_{resp,Q}$, $\mathcal{A}$ outputs $(x,y)$,
then by one-version property, $m_{resp,Q}$ reveals one and only one between $x$ and $y$. (Otherwise, if $m_{resp,Q}$ can reveal another value $v$ other than $x$ and $y$, then we can obtain a successful algorithm which outputs $x, y, v$ given $m_{resp, P}$ and $m_{resp,Q}$, violating one-version property.) By $Q$'s indistinguishability between $E_1$ and $E_2$, in $E_1$, $m_{resp,Q}$ reveals one and only one between $x$ and $y$. W.l.o.g., let $m_{resp,Q}$ reveal $x$. By the construction of $E_{prefix}$, the return value of $ROT$ in $E_1$ cannot include $\bot$. As $C_r$ has not requested any transaction before, then in $E_1$, the return value depends solely on $m_{resp,P}$ and $m_{resp,Q}$. As the client-side algorithm $\mathcal{A}$ is successful, thus $\mathcal{A}$ cannot output a value other than $x$ for object $X$. As a result, $ROT$ returns $x$ in $E_1$. A contradiction to the assumption that if $t_P\in[A, B)$ (which matches $E_1$), then ROT may not return $x$ or $y$.

We prove the second statement also by contradiction. Suppose that in some $E_{imp}$, for some $\tau\in[A, t)$, some $ROT$ such that $t_Q = \tau$ returns $x$ or $y$. By Lemma \ref{lma:atomicity}, $ROT$ returns $(x,y)$. 
Then we construct $E_{old}$ which is the same as $E_{imp}$ except that in $E_{old}$, $ROT$ starts before $B$. In $E_{old}$, let $t_P\in(t_s, B)$ and let $t_Q = \tau$; all messages sent by $P$ after $t_s$ are delayed. Thus $Q$ is unable to distinguish between $E_{old}$ and $E_{imp}$ by the time when $Q$ has sent one message to $C_r$ (for $ROT$). Since $ROT$ returns $(x,y)$ in $E_{imp}$, then $m_{resp,Q}$ reveals $x$ or $y$ in $E_{old}$. 
By the construction of $E_{prefix}$, the return value of $ROT$ in $E_{old}$ cannot include $\bot$. As $C_r$ has not requested any transaction before, then in $E_{old}$, the return value depends solely on $m_{resp,P}$ and $m_{resp,Q}$, which must include $x$ or $y$. A contradiction to the assumption in the statement of the lemma.
\end{proof}

\begin{figure}[!h]
    \centering
    \includegraphics[width=0.5\textwidth]{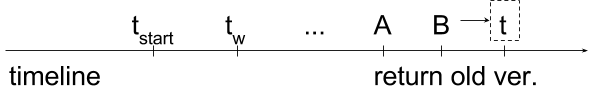}
    \caption{Timeline in Lemma \ref{lma:helper}}
    \label{fig:helper}
\end{figure}

As illustrated in Figure \ref{fig:helper}, Lemma \ref{lma:helper} is based on an assumption that before $B$, old versions are returned (if $ROT$ is appropriately added), shows that $B$ can be prolonged to time $t$. The proof of Lemma \ref{lma:helper} relies on fast read-only transactions.
What remains is the complete proof of Theorem \ref{thm:imp}, which proves Proposition \ref{prop:k},  Proposition \ref{prop:k-2} and  Proposition \ref{prop:base} as well.

\begin{proof}[Proof of Theorem \ref{thm:imp}]
We first prove Proposition \ref{prop:k-2} for any positive $k$ by mathematical induction and then show that $E_{imp}$ indeed violates progress according to Definition \ref{def:visible}. 

By mathematical induction, we start with the base case, i.e., Proposition \ref{prop:base} and Proposition \ref{prop:k-2} for $k = 1$.
Let $A = t_{start}$ and let $B = t_w$. By symmetry, we need only to prove Proposition \ref{prop:base} for $P=P_X$. We show that given $P$, for any $t_P\in [A, B)$, if $C_r$ starts $ROT$ before $t_P$, then $ROT$ may not return $x$ or $y$. For this $ROT$, as $WOT$ has not yet started, $m_{resp, P}$ cannot reveal $x$ or $y$. By one-version property, $m_{resp, P}$ reveals at most one version $v_1$ of $X$ and  $\{m_{resp, P}, m_{resp, Q}\}$ also reveals at most one version $v_2$ of $X$. Therefore $v_1 = v_2\neq x$. As $C_r$ has requested no transaction before, the return value of $ROT$ solely depends on $m_{resp, P}$ and $m_{resp, Q}$. As the client-side algorithm of $C_r$ for the return value of $ROT$ is a successful algorithm, $ROT$ returns $v_1 = v_2\neq x$ for object $X$. (Due to $E_{prefix}$, $ROT$ cannot return $\bot$.) Then by Lemma \ref{lma:atomicity},  $ROT$ may not return $x$ or $y$.
Thus Lemma \ref{lma:helper} applies. As a result, after $B = t_w$, $P$ must send at least one message that precedes some message that arrives at $Q$.
Therefore, Proposition \ref{prop:base} is true for either $P\in\{P_X, P_Y\}$.
Following Proposition \ref{prop:base}, recall the definition of $m_0$ and $m_1$. Let $m_0$ and $m_1$ be sent in $E_{imp}$.
Recall that $T_1$ is the time when the first message preceded by $m_1$ arrives at $D_0$. According to Lemma \ref{lma:helper}, for any $t\in [A, T_1)$, if $C_r$ starts $ROT$ before $t$ and $t_{D_0} = t$, then $ROT$ may not return $x$ or $y$, which proves Proposition \ref{prop:k-2} for $k = 1$.

We continue with the inductive step from case $k$ to case $k+1$.
For our assumption on $k$, let $A = T_{k-1}$, $B = T_k$, $P = D_{k-1}$ and $Q = D_k$. According to the definition of $T_k$, $T_k$ is at least the time when $m_k$ is received. By Proposition \ref{prop:k} for case $k$, $m_k$ is sent at least after $T_{k-1}$. Therefore, $T_k > T_{k-1}$, or $B > A$. Thus Lemma \ref{lma:helper} applies again. As a result, after $T_k$, $D_{k+1} = D_{k-1}$ must send at least one message $m_{k+1}$ which precedes some message that arrives at $D_k$; let $m_{k+1}$ be sent and then for any $t\in[T_{k-1}, T_{k+1})$, if $C_r$ starts $ROT$ such that $T_{D_k} = t$, then $ROT$ may not return $x$ or $y$, which proves  Proposition \ref{prop:k} and  Proposition \ref{prop:k-2} for case $k+1$. Therefore, we conclude Proposition \ref{prop:k-2} for any positive number $k$.

We show that $E_{imp}$ violates progress by contradiction. Suppose that $E_{imp}$ does not violate progress. As there is no other write since the start of $WOT$, then in $E_{imp}$ there is finite time $\tau$ such that any read of object $X$ (or $Y$) which starts at any time $t\geq \tau$ returns $x$ (or $y$).
We have shown that $T_{k+1} > T_{k}$ for any positive $k$. Thus for any finite time $\tau$, there exists $K$ such that for any $k\geq K$, $T_k > \tau$.
Then there exists some $k$ for which $C_r$ starts $ROT$ at some $t \geq \tau$ and $t_{D_{k-1}}$ is less than $T_k$. By Proposition \ref{prop:k-2}, $ROT$ may not return $x$ or $y$. A contradiction. Therefore we find an execution $E_{imp}$ where two values of the same write-only transaction can never be visible, violating progress.
\end{proof}

\section{Proof of Theorem \ref{thm:main}}

In this section, as a proof of Theorem \ref{thm:main}, we show that if some implementation provides invisible read-only transactions, then we reach a contradiction. 
In other words, we show that for every implementation that provides fast read-only transactions, read-only transactions are \emph{visible}.

\subsection{Visible transactions}

From Definition \ref{def:trace}, a read-only transaction $T$ is not invisible if for some client $C$ and $C$'s invocation $I$ of $T$, some execution $E$ (until $I$) can be continued arbitrarily and every execution $E^-$ without $I$ is different from $E$ in addition to the message exchange with $C$ (during the time period of $I$). 
We note that in this case, $T$ does not necessarily leave a trace on the storage. It is possible that for some invocation $I$ of $T$, some execution $E$ (until $I$) can be continued arbitrarily and there is some execution $E^-$ without $I$ which is the same as $E$ except for the message exchange with $C$ (during the time period of $I$). 

This motivates us to define the notion of being visible stronger than that of being not invisible, in Definition \ref{def:strong-visible} below. In Definition \ref{def:strong-visible}, we assume that (1) for each object, some non-$\bot$ value has been visible; (2) the client $C$ which invokes $I$ has not done any operation before $I$; and (3) during $I$, $C$ sends exactly one message $m$ to the servers involved which receive $m$ at the same time, while after the reception of $m$, all servers receive no message before $I$ ends (but still respond to $C$).
We note that even under these assumptions, Definition \ref{def:strong-visible} still shows a strictly stronger notion than being not invisible.

\begin{definition}[Visible transactions]
\label{def:strong-visible}
We say that transaction $T$ is visible if for any invocation $I$ of $T$, any execution $E$ (until $I$) can be continued arbitrarily and every execution $E^-$ without $I$ is different from $E$ in addition to the message exchange with the client which invokes $I$ (during the time period of $I$). 
\end{definition}

Clearly, the definition of visible transactions does not yet quantify the difference between $E$ and $E^-$, or show how much information is exposed by a visible transaction.
In this proof, we quantify the exposed information by proving Proposition \ref{prop:strong}. 
Like Definition \ref{def:strong-visible}, we assume in Proposition \ref{prop:strong} that (1) for each object, some non-$\bot$ value has been visible; (2) the clients $\mathcal{D}$ which invoke $S_T$ have not done any operation before $S_T$; and (3) during $S_T$, every client $C$ in $\mathcal{D}$ sends exactly one message $m$ to the servers involved which receive $m$ at the same time, while after the reception of $m$, all servers receive no message before $S_T$ ends (but still respond to $C$).
The property described in Proposition \ref{prop:strong} is a strictly stronger variant of visible transactions. (To see this, one lets $I\in S_1$ and $I\notin S_2$.) 
Therefore, if we prove Proposition \ref{prop:strong}, then we also prove Theorem \ref{thm:main}. 

\begin{proposition}[Stronger variant of visible transactions]
\label{prop:strong}
Given any causally consistent storage system that provides fast read-only transactions, for some read-only transaction $T$, 
for any set $\mathcal{D}$ of clients and $\mathcal{D}$'s set $S_T$ of concurrent invocations\footnote{Some invocations are said to be concurrent here if the time period between the start and end of these invocations are the same (stronger than the common definition of concurrency).} of $T$, for any subset $S_1\subseteq S_T$, 
any execution $E_1$ where only $S_1$ is invoked (the prefix until $S_1$) can be continued arbitrarily and every execution $E_2$ where only $S_2$ is invoked is different from $E_1$ in addition to the message exchange with $\mathcal{D}$ (during the time period of $S_T$) for any subset $S_2\subseteq S_T$ where $S_2\neq S_1$.
\end{proposition} 

To see that this variant quantifies the exposed information, we count the number of possibilities of these executions that are the same to all clients except for $\mathcal{D}$, with a subset of clients $ss$'s invocations $Inv_{ss}$ of $T$ at the same time where $ss\subseteq \mathcal{D}$. If Proposition \ref{prop:strong} is true for $S_T = Inv_{\mathcal{D}}$, then the number of possibilities is lower-bounded by the number $num$ of subsets of $\mathcal{D}$, implying the amount of difference on the message exchange among these executions.

\subsection{Executions}
Our proof is by contradiction. Suppose that for any read-only transaction $T$, 
for some set $\mathcal{D}$ of clients and $\mathcal{D}$'s set $S_T$ of concurrent invocations of $T$, for some subset $S_1\subseteq S_T$, 
some execution $E_1$ where only $S_1$ is invoked (the prefix until $S_1$) can be continued arbitrarily but still there exists some execution $E_2$ where only $S_2$ is invoked and which is the same as $E_1$ except for the message exchange with $\mathcal{D}$ during the time period of $S_T$ for some subset $S_2\subseteq S_T$ where $S_2\neq S_1$.

We thus construct two executions $E_1$ and $E_2$ following our idea of quantifying information previously. We first recall our construction of the set $S_T$. (We are allowed to do so, as the set is assumed so in the assumption for Proposition \ref{prop:strong}.) Let $S_T$ be the invocations of $ROT=(r(X)*, r(Y)*)$ each of which is invoked by one client in $\mathcal{D}$ at the same time $t_0$. 
Furthermore, we consider $S_T$ that are performed as follows. 
By fast read-only transactions, all messages which a client in $\mathcal{D}$ sends during $ROT$ arrive at $P_X$ and $P_Y$ respectively at the same time. Let $T_1$ denote this time instant and let $T_2$ be the time when $ROT$ eventually ends. 
During $[T_1, T_2]$, by fast read-only transactions, $P_X$ and $P_Y$ receive no message. If there is any such message, they are delayed to at least after $T_2$ but eventually arrive before $\tau_y$.

Now in some $E_1$, only $S_1$ is invoked but every other detail about the execution of $S_T$ above remains the same. We continue $E_1$ with client $C$ performing two writes $w(X)x$ and $w(Y)y$ after $T_2$ to establish $w(X)x\rightsquigarrow w(Y)y$ according to Definition \ref{def:ahamad_causal}. Moreover, after $T_2$, the clients in $\mathcal{D}$ do not invoke any operation. (We are allowed to do so, as $E_1$ can be continued arbitrarily in our assumption for contradiction.)
According to our assumption for contradiction, some $E_2$ is the same as $E_1$ except for the message exchange with $\mathcal{D}$ during the time period of $S_T$, although in $E_2$, only $S_2$ is invoked and $S_2\neq S_1$. Both executions are illustrated in Figure \ref{fig:edi} before the two writes.
In both executions, $y$ is eventually visible. We denote by $\tau$ the time instant after which $y$ is visible in both executions.

\begin{figure}[!h]
    \centering
  \begin{subfigure}[b]{0.44\textwidth}
    \includegraphics[width=\textwidth]{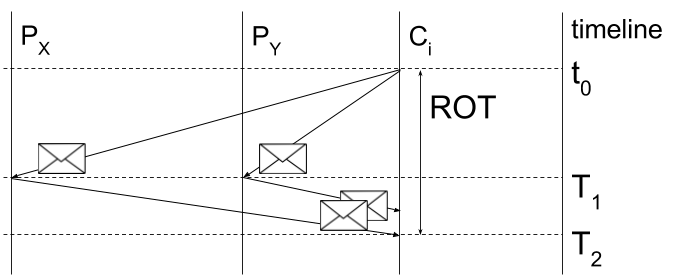}
    \caption{Message schedule of $E_i$}
    \label{fig:edi}
    \end{subfigure}
  \begin{subfigure}[b]{0.54\textwidth}
    \centering
    \includegraphics[width=\textwidth]{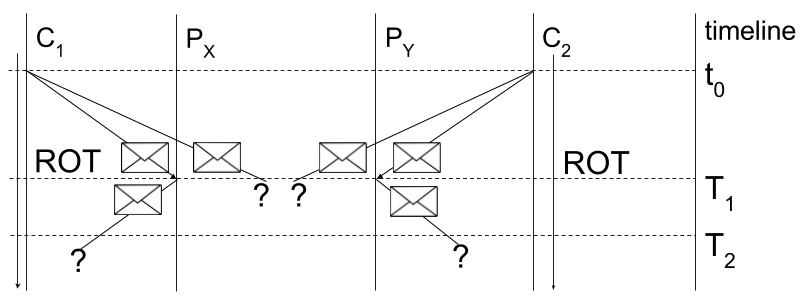}
    \caption{Message schedule of $E_{1,2}$}
    \label{fig:ed1d2}
    \end{subfigure}
    \caption{Construction and extension of $E_i$}
    \label{fig:no-name-2}
\end{figure}

For $i\in\{1,2\}$, let $\mathcal{D}_i$ be the subset of clients which invoke $S_i$ in $E_i$. 
Let $C_i$ be any client in $\mathcal{D}_i$. 
As $E_2$ is the same as $E_1$ except for the message exchange with $\mathcal{D}$ during the time period of $S_T$, w.l.o.g., we assume that $\mathcal{D}_1\backslash\mathcal{D}_2 \neq \emptyset$. We denote by $C_r$ any client in $\mathcal{D}_1\backslash\mathcal{D}_2$ hereafter.

We next construct an execution $E_{1,2}$ based on $E_1$ and $E_2$ to help our proof. Our goal is to let $E_{1,2} = E_1 = E_2$ except for the communication with $\mathcal{D}$ (during the time period of $S_T$) until the same time $\tau$.
In $E_{1,2}$, every client in $\mathcal{D}_1\cup\mathcal{D}_2$ invokes $ROT$ at $t_0$.
As illustrated in Figure \ref{fig:ed1d2}, while every client $C_1\in\mathcal{D}_1$ invokes $ROT$, $P_X$ receives the same message from $C_1$ at the same time $T_1$ and no other message during $[T_1, T_2]$, and sends the same message to $C_1$ at the same time 
as in $E_1$;
$C_1$ sends the same message to $P_Y$ at the same time as in $E_1$, the reception of which may however be delayed by a finite but unbounded amount of time (see below).
Similarly, while every client $C_2\in\mathcal{D}_2$ invokes $ROT$, $P_Y$ receives the same message from $C_2$ at the same time $T_1$ and no other message during $[T_1, T_2]$, and sends the same message to $C_2$ at the same time 
as in $E_2$;
$C_2$ sends the same message to $P_X$ at the same time as in $E_2$, the reception of which may however be delayed by a finite but unbounded amount of time (see below).
For those clients in $\mathcal{D}_1\cap\mathcal{D}_2$, the messages which are said to be possibly delayed still arrive at $T_1$ and follow both the message schedules of $\mathcal{D}_1$ and $\mathcal{D_2}$ above. For the other clients, the messages are indeed delayed by a finite but unbounded amount of time. 
Furthermore, any message which $P_X$ or $P_Y$ sends to a process in $\mathcal{D}$ during $[T_1, T_2]$ is delayed to arrive at least after $T_2$.
Thus by $T_2$, $P_X$ is unable to distinguish between $E_1$ and $E_{1,2}$ while $P_Y$ is unable to distinguish between $E_2$ and $E_{1,2}$. 
As a result, the first message $m_{X,1}$ which $P_X$ sends after $T_2$ in  $E_{1,2}$ is the same message as in $E_1$, and the first message $m_{Y,1}$ which $P_Y$ sends after $T_2$ in  $E_{1,2}$ is the same message as in $E_2$.

According to our assumption for contradiction, $E_1$ and $E_2$ are the same except for the communication with $\mathcal{D}$ during $[t_0, T_2]$. In other words, $E_1$ and $E_2$ are the same regarding the message exchange among servers and message exchange between any server and any client after $T_2$. Therefore, the first message which $P_X$ sends after $T_2$ in $E_2$ is also $m_{X,1}$ and the first message which $P_Y$ sends after $T_2$ in $E_1$ is also $m_{Y,1}$. Therefore, the message exchange among servers in  $E_{1,2}$ continues in the same way as in $E_1$ as well as $E_2$ after $T_2$. Since $\mathcal{D}$ does not invoke any operation after $T_2$ in both executions, then after $T_2$, no client can distinguish between $E_1$ and $E_2$ and therefore the message exchange between any server and any client in  $E_{1,2}$ continues also in the same way as in $E_1$ as well as $E_2$ after $T_2$.
Then even if the delayed messages in  $E_{1,2}$ do not arrive before $\tau$, $E_{1,2} = E_1 = E_2$ except for the communication with $\mathcal{D}$ (during the time period of $S_T$) until the same time $\tau$. We reach our goal as stated previously.

\subsection{Proof}

Our proof starts with the extension of $E_2$ and $E_{1,2}$ since the time instant $\tau$. We show that in a certain extension, $P_Y$ is unable to distinguish between $E_2$ and $E_{1,2}$ and thus returns a new value, which breaks causal consistency. As we reach a contradiction here, we show the correctness of Proposition \ref{prop:strong} as well as that of Theorem \ref{thm:main}.
We also have two remarks on the proof of Proposition \ref{prop:strong}. First, the proof relies on the indistinguishability of servers between executions, implying that fast read-only transactions have to ``write'' to some server to break the indistinguishability (i.e., ``writing'' to a client without the client forwarding the information to a server is not an option).
Second, recall that to quantify the exposed information of read-only transactions, we count the number of possibilities of these executions that are the same to all clients except for $\mathcal{D}$, with a subset $ss$ of $\mathcal{D}$'s invocations $Inv_{ss}$ of $T$ at the same time. Now that the proof shows that Proposition \ref{prop:strong} is indeed true for $S_T = Inv_{\mathcal{D}}$, then the number of possibilities is lower-bounded by the number $2^n$ where $n = |\mathcal{D}|$, implying that each fast read-only transaction in $S_T$ contributes at least one bit in the message exchange.\footnote{The contribution is computed according to the information theory and coding theory. Consider $\mathcal{X}$ as a random variable that takes values in all these $2^n$ executions. Assume that $\mathcal{X}$ takes any one with equal probability. Then the entropy of $\mathcal{X}$ is $n$ bits. According to the coding theory, depending on how the messages exchanged in these executions code $S_T$, each fast read-only transaction may use more than one bits.}

\begin{figure}[!h]
    \centering
  \begin{subfigure}[b]{0.45\textwidth}
    \includegraphics[width=\textwidth]{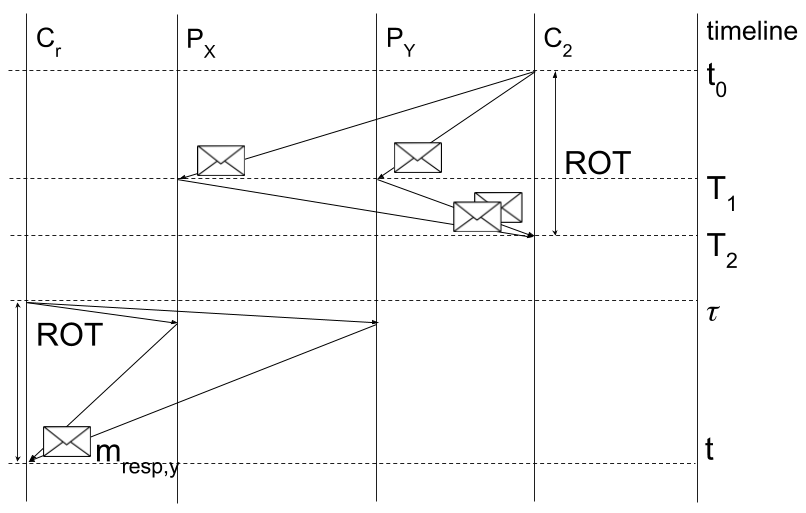}
    \caption{Extension of $E_2$}
    \label{fig:extend-e2}
    \end{subfigure}
    \hfill
  \begin{subfigure}[b]{0.45\textwidth}
    \includegraphics[width=\textwidth]{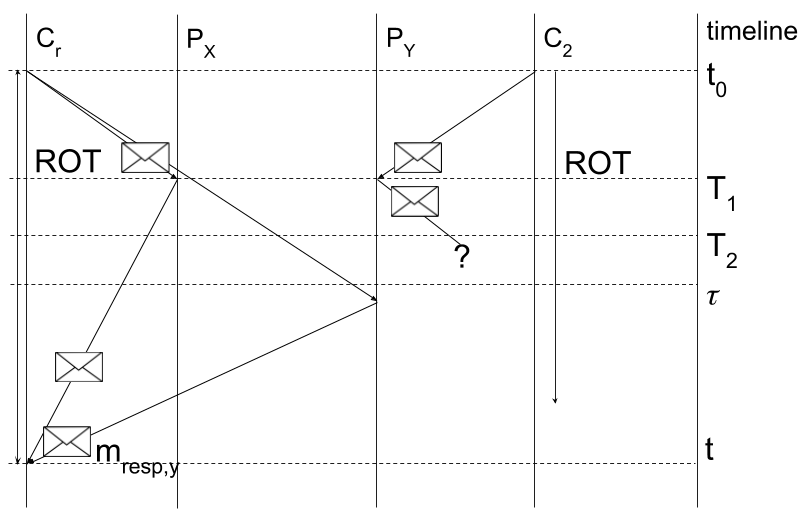}
    \caption{Extension of  $E_{1,2}$.}
    \label{fig:extend-ed1d2}
    \end{subfigure}
    \caption{Extension of two executions}
    \label{fig:extend}
\end{figure}

\begin{proposition}[Contradictory execution]
\label{prop:contradiction}
Execution $E_{1,2}$ can violate causal consistency.
\end{proposition}

\begin{proof}[Proof of Proposition \ref{prop:contradiction}]
We first extend $E_2$ and $E_{1,2}$ after $\tau$. As illustrated in Figure \ref{fig:extend}, we let any client $C_r$ in $\mathcal{D}_1\backslash\mathcal{D}_2$ start $ROT$ immediately after $\tau$ in $E_2$. 
Then in both $E_2$ and $E_{1,2}$, we schedule the message sent from $C_r$ to $P_Y$ during $C_r$'s $ROT$ to arrive at the same time after $\tau$, and by fast read-only transactions and asynchrony, $\exists t$, during $[\tau, t]$, $P_Y$ receives no message but still responds to $C_r$. 
Based on our extension of $E_2$ and $E_{1,2}$, by $t$, $P_Y$ is unable to distinguish between $E_2$ and $E_{1, 2}$.
We also schedule any message which $P_X$ or $P_Y$ sends to $C_r$ arrives at the same time $t$.
Denote the message which $C_r$ receives from $P_Y$ at $t$ by $m_{resp, Y}$, which is thus the same in $E_2$ and $E_{1,2}$. Denote by $m_{resp, X}$, the message which $C_r$ receives from $P_X$ at $t$, which can be different in $E_2$ and $E_{1,2}$.

We now compute the return value of $ROT$ in $E_2$ and $E_{1,2}$. By progress, in $E_2$, $C_r$'s $ROT$ returns $y$ for $Y$. 
By one-version property, $m_{resp, Y}$ reveals exactly one version of $Y$, and $m_{resp, X}$ reveals no version of $Y$. Since $m_{resp, X}$ reveals no version of $Y$, $m_{resp, Y}$ cannot reveal a version of $Y$ different from $y$. In other words, $m_{resp, Y}$ must reveal $y$.
In $E_{1,2}$, $m_{resp, X}$ cannot reveal $x$ as $w(X)x$ starts after $T_2$. Then $m_{resp, X}$ must reveal some value $x^*\neq x$. As $m_{resp, Y}$ has already revealed $y$, messages $\{m_{resp, X}, m_{resp, Y}\}$ cannot reveal other versions of $X$ or $Y$. As in $E_{1,2}$, some values of $X$ and $Y$ have been visible, then the return value cannot be $\bot$; thus the return value of $ROT$ in $E_{1,2}$ is $(x^*, y)$.

We show that the return value $(x^*,y)$ in $E_{1,2}$ violates causal consistency by contradiction. Suppose that $E_{1, 2}$ satisfies causal consistency. Then by Definition \ref{def:causal-consistency}, 
we can totally order all $C_r$'s operations and all write operations such that the last preceding writes of  $X$ and $Y$ before $C_r$'s $ROT$ are $w(X)x^*$ and $w(Y)y$ respectively. Since $x\rightsquigarrow y$, then $w(X)x$ must be ordered before $w(Y)y$. This leads $w(X)x^*$ to be ordered after $w(X)x$. We now extend $E_{1,2}$ so that $C_r$ invokes $ROT_1 = (r(X)*, r(Y)*)$  after $x$ is  visible, which returns value $(x, y)$ by Definition \ref{def:visible}.
According to Definition \ref{def:causal-consistency}, the last preceding write of $X$ before $ROT_1$ must be $w(X)x$. However, $w(X)x^*$ has already been ordered after $w(X)x$ and thus the last preceding write of $X$ before $ROT_1$ is $w(X)x^*$. A contradiction.
\end{proof}

\section{Storage Assumptions}
For presentation simplicity, we made an assumption that servers store disjoint sets of objects. In this section, we show how our results apply to the non-disjoint case. 
A general model of servers' storing objects can be defined as follows. 
Each server still stores a set of objects, but no server stores all objects. For any server $S$, there exists object $o$ such that $S$ does not store $o$.
In this general model, when a client reads or writes some object $o$, the client can possibly request multiple servers all of which store $o$. W.l.o.g., we assume that when client $C$ accesses $o$, $C$ requests all servers that store $o$.

\subsection{Impossibility of fast transactions}
We sketch here why Theorem \ref{thm:imp} still holds in the general model. To prove Theorem \ref{thm:imp}, we still construct a contradictory  execution $E_{imp}$ which, to satisfy causality, contains an infinite number of messages and then violates progress. In $E_{imp}$, $C_w$ issues a write-only transaction $WOT$ which starts at time $t_w$ and writes to all objects. In other words via $WOT$, $C_w$ writes all objects. Since $t_w$, $WOT$ is the only executing transaction. 

The proof is still by induction on the number of messages. (Here $k$ does not denote the number of messages; rather $k$ denotes the number of asynchronous rounds of messages as shown by our inductive step.) We sketch the base case and the inductive step below as the main idea of the proof is the same as in Appendix A. We add an imaginary read-only transaction $ROT$ which reads all objects to $E_{imp}$.
Let $n_{obj}$ be the number of objects read by $ROT$.
For each server $S$, let $m_{resp,S}$ be the message of response from $S$ during $ROT$ if $ROT$ is invoked after all values written by $WOT$ are eventually visible. Let $ss$ be the smallest set of servers such that $\{m_{resp,S}|\forall S\in ss\}$ reveals exactly $n_{obj}$ versions. By one-version property, $ss$ contains at least two servers.
In the base case, we show that after $t_w$, there are at least two servers each of which sends some message that precedes some message which arrives at another server in $ss$. By contradiction. Suppose that at most one server sends such message. 
Since at most one server sends some message that precedes some message which arrives at another server in $ss$, then we assume that one server $R\in ss$ does not send any such message. 
We now make our $ROT$ concrete.
We let the request of $ROT$ arrive earlier than $t_w$ at $R$ and delay the request $ROT$ at $ss\backslash\{R\}$. Based on our assumption on $R$, $ss\backslash\{R\}$ is unable to distinguish between the case where $ROT$ starts before $t_w$ and the case where $ROT$ has not started at all. Then all values written by $WOT$ are eventually visible. After that, we let the request of $WOT$ arrive at $ss\backslash\{R\}$. 
By fast read-only transactions and our assumption for contradiction, after $t_w$, there can be no communication between $R$ and $ss\backslash\{R\}$. Therefore $ss\backslash\{R\}$ returns values written by $WOT$; however $R$ has to return some value not written by $WOT$. We note that $ss\backslash\{R\}\neq\emptyset$ and thus by Lemma \ref{lma:atomicity}, the return value of $ROT$ must break causal consistency. A contradiction.

Therefore we conclude that after $t_w$, there are at least two servers each of which sends some message that precedes some message which arrives at another server in $ss$. Let $S_{base}$ be the set of servers (whether in $ss$ or not) which do so. Let $M_{base}$ be the set of first messages which (1) a server in $S_{base}$ sends and (2) precedes some message that arrives at another server in $ss$.

We now sketch the inductive step. Let $P$ be the first server which receives some message $m_P$ in $M_{base}$. Then from case $k = 1$ to case $k = 2$, we show that after the reception of $m_P$, at least one server sends some message that precedes some message which arrives at another server in $ss$. By contradiction. Suppose that after the reception of $m_P$, no server sends any such message.
We now make our $ROT$ concrete.
We let the request of $ROT$ arrive at one server $R$ in $ss$ before the reception of $m_P$, and let the request of $ROT$ arrive at $ss\backslash\{R\}$ after all values written by $ROT$ are visible. 
Considering the possibility that the request of $ROT$ can arrive at another server before $t_w$, $R$ returns at most one old version for each object $R$ stores. However, by our assumption for contradiction, $ss\backslash\{R\}$ is unable to distinguish this case from the case where $ROT$ starts after all values written by $WOT$ are visible and therefore returns values written by $WOT$. We note that $ss\backslash\{R\}\neq\emptyset$ and thus by Lemma \ref{lma:atomicity}, the return value of $ROT$ must break causal consistency. A contradiction.

Therefore we conclude that after the reception of $m_P$, at least one server sends some message that precedes some message which arrives at another server in $ss$.
Let $S_{2}$ be the set of servers (whether in $ss$ or not) which sends some message that precedes some message which arrives at another server in $ss$ after receiving a message in $M_{base}$. Let $M_{2}$ be the set of first messages which (1) a server in $S_{2}$ sends and (2) precedes some message that arrives at another server in $ss$.

With an abuse of notations, let $P$ be the first server which receives some message $m_P$ in $M_{2}$. For case $k = 3$, we can similarly show that after the reception of $m_P$, at least one server sends some message that precedes some message which arrives at another server in $ss$.
In this way, we add at least one message in each inductive step and also make progress in time which at the end goes to infinity. This completes our construction of $E_{imp}$ as well as the proof sketch of Theorem \ref{thm:imp} in the general model of servers' storing objects.

\subsection{Impossibility of fast invisible transactions}
We sketch here why Theorem \ref{thm:main} still holds in the general model.
To prove Theorem \ref{thm:main}, 
we consider any execution $E_1$ where some client $C_r$ (which has not requested any operation before) starts transaction $ROT$ which reads all objects at the time $t_0$ and before $t_0$, some values of $X$ and $Y$ have been visible. We continue $E_1$ with some client $C$ executing writes which establishes a chain of causal relations. Let $O=\{o_1, o_2,\ldots, o_{n_{obj}}\}$ be the set of all objects. $C$ executes writes $Wr = \{w(o)v|\forall o\in O\}$ so that $\forall k\in\mathbb{Z}$, $2\leq k\leq n_{obj}$, $w(o_{k-1})v_{k-1}\rightsquigarrow w(o_k)v_k$.

Our proof is by contradiction. Suppose that transaction $ROT$ is invisible. Then no matter how $E_1$ is scheduled, there exists some execution $E_2$ such that $E_2$ is the same as $E_1$ except that (1) $C_r$ does not invoke $ROT$, and (2) the message exchange with $C_r$ during the time period of $ROT$ is different. Below we first schedule $E_1$ and then construct $E_{1,2}$. We later show $E_{1,2}$ violates causal consistency.

By fast read-only transactions, we can schedule messages such that the message which $C_1$ sends during $ROT$ arrives at every server respectively at the same time. Let $T_1$ denote this time instant and let $T_2$ be the time when $ROT$ eventually ends. During $[T_1, T_2]$, every server receives no message but still respond to $C_r$.
Clearly all writes of $C$ occur after $T_2$, while $C_r$ does no operation after $T_2$.
All delayed messages eventually arrive before all values written by $C$ can be visible. In $E_1$, we denote the time instant after which all values written are visible by $\tau$.
Next we construct execution $E_{1,2}$ that is indistinguishable from $E_1$ to $P_X$ and from $E_2$ to $P_Y$. The start of $E_{1,2}$ is the same as $E_1$ (as well as $E_2$) until $t_0$.
At $t_0$, $C_r$ still invokes $ROT$.

Before we continue the construction of $E_{1,2}$, we consider an imaginary $ROT$ in $E_2$ which starts after $\tau$. For each server $S$, let $m_{resp,S}$ be the message of response from $S$ during this imaginary $ROT$ if $ROT$ is invoked after $\tau$. Let $ss$ be the smallest set of servers such that $\{m_{resp,S}|\forall S\in ss\}$ reveals exactly $n_{obj}$ versions. By one-version property, there are at least two servers in $ss$. Let $R$ be one server in $ss$ such that $m_{resp,R}$ does not reveal $v_1$.

We now go back to our construction of $E_{1,2}$. We let all servers except for $R$ receive the same message from $C_r$ and send the same message to $C_r$ at the same time as in $E_1$;
$C_r$ sends the same message to $R$ at the same time as in $E_1$, the reception of which is however delayed by a finite but unbounded amount of time. In addition, during $[T_1, T_2]$, all servers receive no message as in $E_1$ (as well as $E_2$).
Thus by $T_2$, all servers except for $R$ are unable to distinguish between $E_1$ and $E_{1,2}$ while $R$ is unable to distinguish between $E_2$ and $E_{1,2}$. Since $E_1 = E_2$ except for the communication with $C_r$ during $[t_0, T_2]$, then $E_{1,2} = E_1 = E_2$ except for the communication with $C_r$ during $[t_0, T_2]$.

Now based on the executions constructed above, we can similarly extend $E_2$ and $E_{1,2}$ after $\tau$. We let $C_r$ start $ROT$ immediately after $\tau$ in $E_2$. In both $E_2$ and $E_{1,2}$, by fast read-only transactions, we schedule the message sent from $C_r$ to $R$ during $C_r$'s $ROT$ to arrive at the same time after $\tau$, and $\exists t$, during $[\tau, t]$, $R$ receives no message but still responds to $C_r$. 
Thus by $t$, $R$ is unable to distinguish between $E_2$ and $E_{1, 2}$.

We next compute the return value of $ROT$ in $E_{1,2}$. By progress and the fact that $R\in ss$, in $E_2$, $R$ returns some new values written by $C$, and then by indistinguishability, in $E_{1,2}$, $R$ returns the same.
However, in $E_{1,2}$, all servers except for $R$ can only return some values which are written before $Wr$. 
W.l.o.g., the return value of $ROT$ in $E_{1,2}$ includes some value $v^*_1\neq v_1$ for object $o_1$ and $v_k$ for some object $o_k$.
According to our assumption, $E_{1,2}$ satisfies causal consistency.
By Definition \ref{def:causal-consistency}, 
we can totally order all $C_r$'s operations and all write operations in $E_{1,2}$ such that the last preceding writes of $o_1$ and $o_k$ before $C_r$'s $ROT$ are $w(o_1)v^*_1$ and $w(o_k)v_k$ respectively. 
This leads $w(o_1)v^*_1$ to be ordered after $w(o_1)v_1$. However, if we extend $E_{1,2}$ so that $C_r$ invokes $ROT_1 = (r(o_1)*, r(o_k)*)$ after $\tau$, then $ROT_1$ returns value $(v_1, v_k)$ and if we do total ordering of $E_{1,2}$ again, then the last preceding write of $o_1$ before $ROT_1$ must be $w(o_1)v_1$, which leads $w(o_1)v_1$ to be ordered after $w(o_1)v^*_1$, a contradiction.
Therefore, we can conclude that $E_{1,2}$ violates causal consistency, which completes our proof sketch of Theorem \ref{thm:main} in the general model.

\section{Alternative Protocols}
For completeness of our discussion in Section 6, we here sketch two implementations, one using asynchronous propagation of information among servers and one assuming the existence of a global accurate block.

\subsection{Visible fast read-only transactions}

We sketch below an algorithm $\mathcal{A}$ for fast read-only transactions. To comply with our Theorem \ref{thm:imp}, we restrict all transactions to be read-only.
The goal of $\mathcal{A}$ is to better understand our Theorem \ref{thm:main}. Theorem \ref{thm:main} shows that fast read-only transactions are visible. The intuition of Theorem \ref{thm:main} is that after a fast read-only transaction $T$, servers may need to communicate the information of $T$ among themselves. However, it is not clear when such communication occurs.
The COPS-SNOW \cite{lu_snow_2016} algorithm shows that the communication can take place during clients' requests of writes. 
Our algorithm $\mathcal{A}$ below shows that the communication can actually take place asynchronously. In addition, while COPS-SNOW guarantees a value to be visible immediately after its write, $\mathcal{A}$ guarantees only eventual visibility; thus a trade-off between the freshness of values and latency perceived by clients is also implied.

We sketch below first the data structure which each process maintains. All processes maintain locally their logical timestamps and update their timestamps whenever they find their local ones lag behind. They also move their logical timestamps forward when some communication with other processes is made. 
Every client additionally maintains the causal dependencies of the current operation (i.e., the operations each of which causally precedes the current one), called context. The maintenance of context is done in a similar way as COPS \cite{lloyd_settle_2011} and COPS-SNOW \cite{lu_snow_2016}. 
Every server needs to store the causal dependencies passed as an argument of some client' write.
Every server additionally maintains a data structure called $OldTx$ for each object stored. 
We next sketch how writes and read-only transactions are handled.
\begin{itemize}
\item Every client sends its logical timestamp as well as context when requesting a write. A server stores the value written along with the server's updated logical timestamp, causal dependencies, and returns to the client.
\item Every client $C$ sends its logical timestamp as well as context when requesting a read-only transaction $tx$. A server first searches $tx$ in $OldTx$, and returns a pre-computed value according to entry $tx$ in $OldTx$ if $tx\in OldTx$. Otherwise, a server returns some value already observed by $C$ (in its context) or some value marked as ``visible''.
\end{itemize} 
We finally sketch how $OldTx$ is maintained and communicated (where asynchronous propagation mentioned in Section 6 takes place).
\begin{itemize}
\item After a server $S$ responds to a client's write request of value $w$, $S$ sends a request to every server which stores some value $v$ such that $v\rightsquigarrow w$. Any server responds such request with its local $OldTx$ when $v$ is marked as ``visible''.
\item After $S$ receives a response from all servers which store some value that causally precedes $w$, $S$ stores their $OldTx$s into $S$'s local one,  chooses a value already observed by the client of $tx$ or a value $w^*$ which is written before $w$\footnote{In order to choose a value correctly, in the algorithm, $S$ actually sends a request after all values written before $w$ are marked as ``visible''. Also, $S$ does not choose a value for some $tx$ which $S$ has chosen before (which can happen when some value written before $w$ is marked as ``visible''). In this way, $S$ can choose $w^*$ as the last value written before $w$.} for each transaction $tx$ in $OldTx$ and marks $w$ as ``visible''.
\end{itemize} 
Any read-only transaction is stored and marked as ``current'' during its operation at any server. A ``current'' transaction $T$ is put in $OldTx$ when some value $w$ is ``visible'' and $T$ has returned a value written before $w$ of the same object.\\

\noindent \textbf{Proof sketch of Correctness.} Our algorithm $\mathcal{A}$ above provides fast read-only transactions.  As every message eventually arrives at its destination (and therefore asynchronous propagation eventually ends), $\mathcal{A}$ satisfies progress. We can show that $\mathcal{A}$ satisfies causal consistency by contradiction. Suppose that some execution $E$ violates causal consistency. Then $E$ includes at least one read-only transaction. Assume that in $E$, for some client $C$, the ordering of all writes and $C$'s transactions breaks causal consistency.
Clearly, without any read-only transaction, we can order all writes in a way that respects causality. In addition, we can also order all writes of the same object according to the increasing timestamps of these writes and still respect causality. (We call the ordering of writes of the same object according to the timestamps by object relation. In addition, we say that two writes $w_1\rightarrow w_2$, if (1) $w_1$ is before $w_2$ by object relation or by causal relation, or (2) $\exists$ some write $w_3$ such that $w_1\rightarrow w_3$ and $w_3 \rightarrow w_2$.)
Let $to$ be any such ordering.
We then add $C$'s read-only transactions on $to$ one by one. Let $T$ be the first read-only transaction such that some $to_1$ exists which can include $C$'s read-only transactions before $T$ but for any $to$, $T$ as well as $C$'s read-only transactions before $T$ cannot be placed in $to$ property (i.e., to satisfy causal serialization.)

Let $A$ be the set of such ordering $to$ that  can include $C$'s read-only transactions before $T$ and let $to_1$ be any ordering in $A$.
First, by our algorithm $\mathcal{A}$, it is easy to verify the following property of $to_1$: for any two reads $r(a)u, r(o)v^*\in T$, if in $to_1$, $\exists w(a)u^*$ such that $w(a)u$ is before $w(a)u^*$ and $w(a)u^*$ is before $w(o)v^*$, then $w(a)u^*\rightarrow w(o)v^*$ does not hold.
Second, based on the property and $to_1$, we construct $to_2$ as follows.
For any $r(a)u\in T$, consider $w(a)u^*$ as the first write of $a$ such that (1) $w(a)u^*$ is after $w(a)u$, (2) some $w(o)v^*$ is after $w(a)u^*$ and (3) $r(o)v^*\in T$.
We let $W_u$ be the set of such write $w(o)v^*$ that is after $w(a)u^*$ and (3) $r(o)v^*\in T$.
We then augment $W_u$ by adding the precedence of each element according to relation $\rightarrow$, and we do this until no more write after $w(a)u^*$ in $to_1$ can be added.
Let $ss$ be the subsequence of $to_1$ which contains all writes in $W_u$. 
We move $ss$ immediately before $w(a)u^*$.
Below we verify that the resulting ordering $to_m$ (not yet our goal $to_2$) falls in $A$.
By the construction based on relation $\rightarrow$, $to_m$ still respects causality and orders all writes of the same object according to the timestamps of these writes. 
We also verify that $C$'s read-only transactions before $T$ can be placed in $to_m$ by contradiction: suppose that some read-only transaction $T_0$ before $T$ finds the last preceding write of $T_0$ incorrect. As a result, $T_0$ must be after $w(a)u^*$ back in $to_1$; then $r(a)u^*\in T_0$; however, as $T$ returns a value at least observed by $C$'s previous operations, $T$ cannot return $u$ when $T_0$ has returned $u^*$, which gives a contradiction.
Now that the move of $ss$ creates no new pair $w(a)u$ and $w(o)v^*$ such that $r(a)u, r(o)v^*\in T$ and $w(o)v^*$ is after $w(a)u^*$ and $w(a)u^*$ is after $w(a)u$, then after a finite number of moves, we can construct an ordering $to_2\in A$ such that for any $r(a)u\in T$, $W_u = \emptyset$. 
Finally, if we place $T$ after the last write that corresponds some read in $T$ in $to_2$, then we find all preceding writes of $T$ are correct, a contradiction of our assumption. As a result, we must conclude that $\mathcal{A}$ satisfies causal consistency.

\subsection{Timestamp-based implementation}
The  algorithm here relies on the assumption that all processes can access a global accurate clock and accurate timestamps:
\begin{itemize}
\item Before any client starts a transaction, the client accesses the clock and stamps the transaction with the current time;
\item Every client sends the accurate timestamp while requesting a transaction;
\item If an operation writes a value to an object, then the server that stores the object attaches the timestamp to the value;
\item If an operation reads a value of an object, then the server that stores the object returns the value with the highest timestamp which is still smaller than the timestamp stamped by the client of the transaction.
\end{itemize}
Each transaction induces one communication round and is invisible.
The algorithm guarantees progress as the clock makes progress.

If updates are only allowed outside transactions, then the algorithm satisfies causal consistency trivially as the accurate timestamp serializes all these individual writes, The algorithm thus circumvents our Theorem \ref{thm:main} no matter whether communication delays are bounded or not.

If general transactions are allowed, then the algorithm can be adapted to still satisfy causal consistency when the message delay is upper-bounded by time $u$. More specifically, a client imposes that every transaction is executed for time $2u$ and instead of comparing with the timestamp $ts$ stamped by the client $C$, the server compares the timestamp of a value with $ts - 2u$ when responding to a read. All writes are still serialized, and these writes linked within the same transaction can be serialized at the same time.
The algorithm thus circumvents our Theorem \ref{thm:imp} when communication delays are bounded but a global accurate clock is accessible.

\end{appendix}

\end{document}